\documentclass[amsmath,amssymb,reprint, longbibliography]{revtex4-1}

\expandafter\let\csname equation*\endcsname\relax
\expandafter\let\csname endequation*\endcsname\relax

\usepackage{amsthm}
\usepackage{latexsym}
\usepackage{slashed}
\usepackage{mathrsfs}
\usepackage{eucal}
\usepackage{tipa}
\usepackage{graphicx}
\usepackage{hyperref}
\usepackage{tikz}

\usepackage{mathtools}

\usepackage[utf8]{inputenc}
\usepackage[T1]{fontenc}
\usepackage{mathptmx}

\theoremstyle{plain}
\newtheorem{theorem}{Theorem}

\newtheorem{lemma}[theorem]{Lemma}
\newtheorem{corollary}[theorem]{Corollary}

\newtheorem{definition}[theorem]{Definition}
\newtheorem{remark}[theorem]{Remark}

\newcommand{\hide}[1]{}
\renewcommand{\Re}{\text{Re}}
\renewcommand{\Im}{\text{Im}}

\def\MindexA{\mathcal{A}}

\def\sDiv{\mathscr{D}}
\def\sCurl{\mathscr{C}}
\def\sCurlDagger{\mathscr{C}^\dagger}
\def\sTwist{\mathscr{T}}
\DeclareMathOperator{\tho}{\text{\rm{\textthorn}}}
\DeclareMathOperator{\edt}{\eth}

\def\TMEEop{\mathbf{E}}
\def\TMETop{\mathbf{T}}
\def\TMESop{\mathbf{S}}
\def\TSISop{\widehat{\mathbf{S}}}
\def\TMEOop{\mathbf{O}}
\def\TSIOop{\widehat{\mathbf{O}}}
\def\TMEFop{\mathbf{F}}
\def\TSIFop{\widehat{\mathbf{F}}}
\def\TSIUop{\widehat{\mathbf{L}}}
\def\TSINop{\widehat{\mathbf{N}}}
\def\SymSpin{\mathcal{S}}
\def\defeq{\coloneqq}

\begin{document}

\title[Symmetries of linearized gravity from adjoint operators]{Symmetries of linearized gravity from adjoint operators}

\author{Steffen Aksteiner} \email{steffen.aksteiner@aei.mpg.de}
\affiliation{Albert Einstein Institute, Am M\"uhlenberg 1, D-14476 Potsdam,
Germany}

\author{Thomas B\"ackdahl} \email{thomas.backdahl@chalmers.se}
\affiliation{Mathematical Sciences, Chalmers University of Technology and University of Gothenburg, SE-412 96 Gothenburg, Sweden}

\date{\today}

\begin{abstract}
Based on operator identities and their formal adjoints, we derive two symmetry operators for the linearized Einstein operator on vacuum backgrounds of Petrov type D and in particular the Kerr spacetime. One of them is of differential order four and coincides with a result of Cohen and Kegeles. The other one is a new operator of differential order six. The corresponding operator identities are based on the Teukolsky equation and the Teukolsky-Starobinski identities, respectively. The method applies to other field equations as well, which is illustrated with the Maxwell equation. The resulting symmetry operators are connected to Hertz and Debye potentials as well as to the separability of the Teukolsky equation for both Maxwell and linearized gravity.
\end{abstract}

\maketitle

\section{Introduction}

A remarkable property of the Kerr black hole spacetime is that certain wave equations derived from the spin-$s$ field equations are separable. Geometrically this cannot be explained by the two isometries of the Kerr solution alone, but also requires its Killing-Yano tensor. The connection to separability is made by partial differential operators involving the Killing vectors and Killing-Yano tensor such that solutions of a field equation are mapped to solutions. We call such operators \textit{symmetry operators}. 

Second order symmetry operators for the scalar wave equation on Kerr were first discussed by Carter, \cite{Carter:1968}. The spin-1/2 case was analyzed by Carter and McLenaghan \cite{Carter:McLenaghan:1979} and Kamran and McLenaghan \cite{Kamran:McLenaghan:1984} and for spin-1 Kalnins and Miller \cite{kalnins:miller:1989} constructed two second order symmetry operators, see also \cite{Carter:2006} for an overview. Recently symmetry operators up to second order for scalar wave, massless spin-$1/2$ and Maxwell equations on 4-dimensional Lorentzian spacetimes were classified in~\cite{ABB:symop:2014CQGra..31m5015A}.

For spin-2, linearized gravity, only partial results are known. Cohen and Kegeles \cite{1979PhRvD..19.1641K} constructed Hertz potentials for spin-2, which is a covariant form of the scalar Debye potential construction presented earlier by Cohen and Kegeles \cite{1975PhLA...54....5C}, Chrzanowski \cite{Chrzanowski:1975} and also Wald \cite{wald:1978PhRvL..41..203W}. Both can be thought of as higher order versions of the well known vector potential construction for Maxwell fields and they can be interpreted as symmetry operators. In this paper we use a method introduced by Wald \cite{wald:1978PhRvL..41..203W} to re-derive the covariant result of Cohen and Kegeles. Furthermore, we show that a slight extension of Wald's adjoint method can be used to construct a new symmetry operator of order six.

We start by reviewing symmetry operators for the scalar wave equation $\square \phi = 0$ on a Kerr background. The two isometries generated by say $\xi^a, \zeta^a$, correspond to time translation and angular rotation, directly lead to first order (in this case even commuting) symmetry operators via Lie dragging,
\begin{align}
[\mathcal{L}_\xi, \square] =0 && [\mathcal{L}_\zeta,\square] = 0.
\end{align}
Moreover, as Carter showed \cite{carter:1977PhRvD..16.3395C} the operator $Q = \nabla^a K_{ab} \nabla^b$, with $K_{ab}$ a Killing tensor leads to
\begin{align}\label{eq:ScalarCarterCOmmutator}
[Q, \square] =0,
\end{align}
i.e. it is a symmetry operator not reducible to compositions of Lie derivatives along isometries. It is this operator which encodes the separability of the $r$ and $\theta$ coordinates on the Kerr spacetime in Boyer-Lindquist coordinates. However, we note that $Q$ is neither \textit{purely radial} in the sense that $Q = \sum_i f^{ij}(r) \partial_i \partial_j$ nor \textit{purely angular} in the sense that $Q = \sum_i f^{ij}(\theta) \partial_i \partial_j$. Only by adding/subtracting terms vanishing due to the field equation can this form be accomplished and we find a similar structure for the Maxwell field and linearized gravity.

As was shown by Andersson and Blue \cite{andersson:blue:0908.2265}, symmetry operators play an important role in proving decay of scalar waves on a Kerr background. This decay result can be seen both as tool and as a model problem for the black hole stability problem. An essential step in their proof was to prove integrated local energy estimates or Morawetz estimates. For this to work, they needed to identify trapped modes, i.e. waves following the orbiting null geodesics of the Kerr spacetime and make the estimate insensitive to this phenomena. This could be done using the symmetry operators and a direct relation between them and the constants of motion for the geodesic equation. Furthermore, by inserting symmetry operators into energy estimates one can easily increase the differential order of the estimates, which is needed to prove point wise estimates.

We expect that similar techniques could be used also for higher spin fields, like Maxwell fields and linearized gravity.
We therefore need a large set of symmetry operators. Lie derivatives along isometries are obvious first order symmetry operators and we are not discussing them and their compositions further. Apart from these there are only two second order symmetry operators for Maxwell constructed from the Killing-Yano tensor, see~\cite{ABB:symop:2014CQGra..31m5015A}.

It is the goal of this paper to derive analogous symmetry operators for the linearized Einstein operator itself on vacuum spacetimes of Petrov type D. The method we use is based on operator identities originally proposed by Wald \cite{wald:1978PhRvL..41..203W}. An \textit{operator identity} is an equality of operators which requires only commutators and no field equations for its verification, hence it is a property of the background and not the particular field theory we study on it. It should be pointed out that the method is constructive, but we cannot say anything about completeness of the generated set of symmetry operators yet. Suppose $\TMEEop$ is a formally self-adjoint differential operator defining the field equation
\begin{align*}
 \TMEEop f = 0
\end{align*}
under consideration. Applying a differential operator $\TMESop$ on this gives an integrability condition $ \TMESop \TMEEop f = 0$. For certain operators $\TMESop$ there is an alternative form of the composed operator $\TMESop\TMEEop$ satisfying an operator identity $ \TMESop\TMEEop = \TMEOop\TMETop$, such that the operator $\TMEOop$ has particularly nice properties (in the cases discussed in this paper, it is formally self-adjoint). From this set-up a symmetry operator for $\TMEEop$ follows from
\begin{theorem}[\cite{wald:1978PhRvL..41..203W}] \label{thm:AdjointMethod}
Suppose the identity 
\begin{align} \label{eq:adjoint} 
\TMESop\TMEEop = \TMEOop\TMETop,
\end{align}
holds for linear partial differential operators $\TMESop, \TMEEop, \TMEOop$ and $\TMETop$. 
Suppose $\psi$ satisfies $\TMEOop^\dagger \psi = 0$, where $\dagger$ denotes the adjoint of the operator with respect to some inner product. Then $\TMESop^\dagger \psi$ satisfies $\TMEEop^\dagger \TMESop^\dagger \psi = 0$. Thus, in particular, if $\TMEEop$ is self-adjoint then $\phi = \TMESop^\dagger \psi$ is a solution of $\TMEEop \phi = 0$.
\end{theorem}
\begin{proof}
The adjoint of \eqref{eq:adjoint} is given by $\TMEEop^\dagger\TMESop^\dagger = \TMETop^\dagger\TMEOop^\dagger$. Applied to $\psi$, we obtain $\TMEEop^\dagger \TMESop^\dagger \psi = 0$.
\end{proof}
We find that there are two integrability conditions, for Maxwell and linearized gravity, which fit into this scheme. The first one is the second order Teukolsky master equation (TME), a wave equation for certain field strength components of the Maxwell field or certain curvature components of the linearized gravitational field. We denote the TME operator by $\TMEOop$ so that its operator identity is given by \eqref{eq:adjoint}. The second one is the Teukolsky-Starobinski identity (TSI), a differential relation between the aforementioned components. It is of differential order two for the Maxwell equation and of order four for linearized gravity. We introduce hats to distinguish from the first case and write $\TSISop\TMEEop = \TSIOop\TMETop$ with $\TSIOop$ the TSI operator. For linearized gravity on a curved background an additional term occurs in the operator identity and we refer to section~\ref{sec:Spin2TSI} for details. The covariant form of the spin-1 TME and TSI can be found in \cite{2015arXiv150402069A} and for linearized gravity its full form was first derived in \cite{2016arXiv160106084A}.

\begin{remark}
In the original work \cite{wald:1978PhRvL..41..203W}, Wald applied the theorem to the TME for Maxwell and linearized gravity in a Newman-Penrose component form and made two key observations:
\begin{enumerate}
 \item \textit{Connecting the TME to covariant field equations}\\
  First use the (self-adjoint) wave equations for the vector potential (spin-1) or the linearized metric (spin-2) to define the $\TMEEop$ operator covariantly. Then there exists an operator $\TMESop$ connected to the TME operator $\TMEOop$ via \eqref{eq:adjoint}. In fact, the operator $\TMESop$ can be read-off from the source terms of the TME in Teukolsky's original work.
 \item \textit{Adjointness property of TME component operator}\\
 The TME operator (component) for one extreme scalar (say $\phi_0$ for spin-1 or $\dot{\Psi}_0$ for spin-2) is, up to rescaling by a scalar field, the adjoint of the TME operator for the other extreme scalar ($\phi_2$ for spin-1 or $\dot{\Psi}_4$ for spin-2).
\end{enumerate}
Identifying $\TMEOop^\dagger$ with the Debye potential operator, Wald provided a re-derivation of the spin-1 Debye potential formulation of Cohen and Kegeles \cite{cohen:kegeles:1974PhRvD..10.1070C} and the first complete proof of the spin-2 Debye potential formulation initiated by Cohen and Kegeles \cite{1975PhLA...54....5C} and Chrzanowski \cite{Chrzanowski:1975}. See also \cite{Whiting:2005hr} for an overview.
\end{remark} 
We note that the second observation mentioned in the remark is naturally encoded in the self-adjointness of the \textit{covariant} TME operator, as we will see below. Therefore the Debye potential formulation can also be discussed in terms of symmetry operators, see sections~\ref{sec:Spin1Interpretation} and \ref{sec:Spin2Interpretation}.

For the Maxwell case we find that the symmetry operator following from the TSI operator identity is closely related to the separability of the TME operator (from that point of view, it is the generalization of the Carter operator). 
To see this one needs to make use of the freedom to add/subtract terms vanishing due to the field equations to produce symmetry operators which are "purely angular" or "purely radial" in the sense discussed on page~\pageref{eq:ScalarCarterCOmmutator}.
The symmetry operator following from the TME operator identity is equivalent a Hertz potential construction. See also \cite[Chapter 5]{aksteiner:thesis}. Both symmetry operators are of differential order two, and they coincide with the irreducible symmetry operators found in the complete classification \cite{ABB:symop:2014CQGra..31m5015A}, see remark~\ref{rem:Spin1RelToClassification}.

The two symmetry operators for linearized gravity are of higher order than in the spin-1 case.
The one we get from the TSI operator identity is of order six. Since the TME operator is of order two, the relation to separability is not as immediate as for the spin-1 case. See however lemma~\ref{lem:Spin21stkindComps} and remark~\ref{rem:Spin2TMESymop}. 
From the TME operator identity we get a fourth order operator, which is equivalent to the covariant Hertz potential formalism of Cohen and Kegeles \cite{1979PhRvD..19.1641K}. 

We use Penrose's 2-spinor formalism and irreducible decompositions leading to symmetric spinors exclusively. The main information about the type D geometry is encoded in the existence of a Killing spinor $\kappa_{AB}$ satisfying $\nabla_{A'(A}\kappa_{BC)}=0$. On the Kerr spacetime the Killing spinor contains the same information as the Killing-Yano tensor. We therefore construct covariant operators involving the Killing spinor which allows us to express all symmetry operators covariantly.

All calculations in this paper were performed in the \emph{xAct} \cite{xAct} suit for \emph{Mathematica}, and in particular we have used and developed the \emph{SymManipulator} and \emph{SpinFrames} packages for this work. The typeset ready equations were produced with \emph{TexAct}.

\subsection*{Overview}
In the preliminaries section~\ref{sec:Preliminaries} we introduce the inner products for the adjoint method and a set of algebraic and differential operators tied to the Petrov type D geometry. In section~\ref{sec:spin1} we analyze the source-free Maxwell equation (spin-1) in covariant spinorial form. The self-adjointness of TME and TSI is shown and afterwards used in  Theorem~\ref{thm:Spin1TMETSIOperatorIdentity} to construct two second order symmetry operators.
In subsection~\ref{sec:Spin1Interpretation} the symmetry operators for spin-1 are related to the covariant and GHP component forms of  Hertz and Debye potentials and to the separability of the TME. In section~\ref{sec:Spin2} we formulate the linearized gravity equations in terms of suitable operators. The TME and TSI are then separately treated in subsections~\ref{sec:Spin2TME} and \ref{sec:Spin2TSI}, respectively. In each case formal self-adjointness is shown and used to construct a symmetry operator. These are the main results of the paper and stated in Theorems \ref{thm:Spin2TMEOperatorIdentity} and \ref{thm:Spin2TSIOperatorIdentity}, respectively. In subsection~\ref{sec:Spin2Interpretation} the symmetry operators for spin-2 are related to the covariant and GHP component forms of  Hertz and Debye potentials and to the separability of the TME. In the concluding section~\ref{sec:conclusion} we summarize the results and also point to interesting directions for the future. In appendix~\ref{sec:adjoints}, we present the formal adjoints of the operators introduced in this paper. Appendix~\ref{app:KopComm} contains a list of commutator relations for some of these operators. In appendix~\ref{sec:GHPform} we list the GHP component form of selected operators.

\section{Preliminaries} \label{sec:Preliminaries}
We use the 2-spinor formalism and sign convention of Penrose \& Rindler \cite{Penrose:1986fk} in which $24\Lambda$ denotes the Ricci scalar, $\Phi_{ABA'B'}$ is the trace-free Ricci spinor and $\Psi_{ABCD}$ is the Weyl spinor. All spinors can be decomposed into irreducible parts, which are symmetric, so we can work with symmetric spinors exclusively. Throughout the paper we will therefore assume that all spinors are symmetric unless explicitly stated otherwise. We denote the space of symmetric spinors of valence $(k,l)$ by $\SymSpin_{k,l}$. Sets of spinor indices are collected into a multi index, e.g. $\MindexA_k \defeq A_1 \dots A_k$, or suppressed completely in cases where it does not lead to confusion.

Anti-self-dual and self-dual 2-forms  are equivalent to elements of $\SymSpin_{2,0}$ and $\SymSpin_{0,2}$ respectively. Anti-self-dual and self-dual 4-tensors with Weyl symmetries correspond to elements of $\SymSpin_{4,0}$ and $\SymSpin_{0,4}$ respectively. To avoid clutter in the notation, we use the same symbol for corresponding operators (e.g. $\TMEEop$) in the spin-1 and 2 cases in sections~\ref{sec:spin1} and \ref{sec:Spin2}, respectively.

Component expressions with respect to a spinor dyad $(o_A,\iota_A)$ will be expressed in the GHP notation. For certain operators in GHP form we will talk about "purely angular" or "purely radial" operators. By this we mean that after separation of variables the "purely angular" operators will only involve angular variables, and the "purely radial" operators will only involve radial variables.

Following \cite[sec. 2.2]{2015arXiv150402069A}, we denote the formal adjoint of a linear differential operator $\mathbf{A}$ with respect to the bilinear pairing
\begin{align} \label{eq:InnerProd1}
(\phi_{\MindexA_k \MindexA'_l},\psi_{\MindexA_k \MindexA'_l}) = \int \phi_{\MindexA_k \MindexA'_l} \psi^{\MindexA_k \MindexA'_l} d\mu
\end{align}
by $\mathbf{A}^\dagger$, and the adjoint with respect to the sesquilinear pairing
\begin{align} \label{eq:InnerProd2} 
\langle\phi_{\MindexA_k \MindexA'_l},\psi_{\MindexA_l \MindexA'_k}\rangle = \int \phi_{\MindexA_k \MindexA'_l} \overline{\psi}^{\MindexA_k \MindexA'_l} d\mu
\end{align}
by $\mathbf{A}^\star$. We refer to them as $\dagger$-adjoint and $\star$-adjoint, respectively. The integrations are formal and with respect to the geometric volume element $ d\mu$ of some background geometry. The adjoints of certain natural operators with respect to both inner products are listed in appendix \ref{sec:adjoints}. Also note that $\overline{\mathbf{A}^\dagger(\phi)}~=~\mathbf{A}^\star(\bar{\phi})$. For  $\mathbf{A}:\SymSpin_{k,l} \to \SymSpin_{u,v}$ the adjoint operators map according to $\mathbf{A}^\dagger:\SymSpin_{u,v} \to \SymSpin_{k,l}$ and $\mathbf{A}^\star:\SymSpin_{v,u} \to \SymSpin_{l,k}$. The adjoint operator argument in Theorem~\ref{thm:AdjointMethod} will be used for both the $\dagger$-adjoint and the $\star$-adjoint.

We consider vacuum type D spaces, and collect certain algebraic and differential operators on such spaces which were first introduced (with examples) in \cite{2016arXiv160106084A}. The fundamental operators $\sCurl, \sCurlDagger, \sTwist, \sDiv$ acting on $\SymSpin_{k,l}$ are defined as the irreducible parts of the covariant derivative $\nabla_{AA'} \varphi_{BC\cdots DB'C'\cdots D'}$ of the symmetric spinor $\varphi \in \SymSpin_{k,l}$. See \cite{ABB:symop:2014CQGra..31m5015A} for a detailed discussion of their properties including commutators. The main feature of the Petrov type D geometry is encoded in the symmetric Killing spinor $\kappa_{AB}$ found in \cite{walker:penrose:1970CMaPh..18..265W}, satisfying 
\begin{align} 
\nabla_{(A|A'|}\kappa_{BC)} = 0.
\end{align} 
In a principal dyad the Killing spinor takes the simple form 
\begin{align} \label{eq:TypeDKS}
\kappa_{AB}={}&-2 \kappa_1 o_{(A}\iota_{B)},
\end{align}
with $\kappa_1 \propto \Psi_2^{-1/3}$ and $\Psi_2$ being the only non-vanishing component of the Weyl spinor. The constant factor of proportionality can be choosen arbitrarily. 
However, see Lemma~\ref{lem:Spin21stkindComps} for an explicit form in the Schwarzschild geometry. Note that $\kappa_1$ and $\Psi_2$ can be expressed covariantly via the relations $\kappa_{AB} \kappa^{AB} = -2 \kappa_{1}{}^2$ and $\Psi_{ABCD} \Psi^{ABCD}=6 \Psi_{2}^2$. Hence, we can allow $\kappa_1$ and $\Psi_2$ in covariant expressions. From a commutator it follows that
\begin{align}\label{eq:XiDef}
\xi_{AA'}={}&- \nabla_{BA'}\kappa_{A}{}^{B}, 
\end{align}
is a Killing vector field. Another important vector field is defined by
\begin{align} \label{eq:UU11Def}
U_{AA'}={}&- \frac{\kappa_{AB} \xi^{B}{}_{A'}}{3 \kappa_1^2} = - \nabla_{AA'}\log(\kappa_1) .
\end{align}
Setting $\nabla_{(n,m)\, AA'} \defeq \nabla_{AA'} + n U_{AA'} + m \bar U_{AA'}$ for constants $n,m$, we define the extended fundamental operators acting on $\varphi \in \SymSpin_{k,l}$ by,
\begin{subequations}\label{eq:ExtendedFundSpinOp}
\begin{align}
(\sDiv_{(n,m)} \varphi)_{\MindexA_{k-1}}{}^{\MindexA'_{l-1}}
&\defeq
\nabla_{(n,m)}^{BB'} \varphi_{\MindexA_{k-1} B}{}^{\MindexA'_{l-1}}{}_{B'} ,  
\\
(\sCurl_{(n,m)} \varphi)_{\MindexA_{k+1}}{}^{\MindexA'_{l-1}}
&\defeq 
\nabla_{(n,m)\, (A_1}{}^{B'} \varphi_{A_2 \dots A_{k+1})}{}^{\MindexA'_{l-1}}{}_{B'} , 
\\
(\sCurlDagger_{(n,m)} \varphi)_{\MindexA_{k-1}}{}^{\MindexA'_{l+1}}
&\defeq 
\nabla_{(n,m)}^{B(A'_1} \varphi_{\MindexA_{k-1} B}{}^{A'_2 \dots A'_{l+1})} , 
\\
(\sTwist_{(n,m)} \varphi)_{\MindexA_{k+1}}{}^{\MindexA'_{l+1}}
&\defeq 
\nabla_{(n,m)\, (A_1}{}^{(A'_1} \varphi_{A_2 \dots A_{k+1})}{}^{A'_2 \dots A'_{l+1})} .
\end{align}
\end{subequations}
For $n = m = 0$ they coincide with the fundamental operators of \cite{ABB:symop:2014CQGra..31m5015A}, and the indices will be suppressed in that case.  Because $U_{AA'}$ is a logarithmic derivative we have
\begin{align} 
 \sCurl_{(n,m)} &= \kappa_{1}^{n}\bar\kappa_{1'}^{m} (\sCurl \kappa_{1}^{-n}\bar\kappa_{1'}^{-m} \bullet),
\end{align}
and similarly for the other operators. In particular it follows that the commutator of extended fundamental spinor operators with $n_1 = n_2, m_1 = m_2$ reduces to the commutator of the fundamental spinor operators given in \cite[Lemma 18]{ABB:symop:2014CQGra..31m5015A}. We can also use this to commute factors of $\kappa_1$ or $\bar\kappa_{1'}$ in or out of the extended fundamental spinor operators like 
\begin{equation} \label{eq:ExtendedFundSpinOpRescaling}
\kappa_{1}^{n}\bar\kappa_{1'}^{m} (\sCurl_{(p,q)} \bullet) = (\sCurl_{(p+n,q+m)} \kappa_{1}^{n}\bar\kappa_{1'}^{m} \bullet),
\end{equation}
and similarly for the other operators. Because $\kappa_1 \propto \Psi_2^{-1/3}$ we also find 
\begin{equation}
\Psi_2^{n} (\sCurl_{(p,q)} \bullet) =(\sCurl_{(p-3n,q)} \Psi_2^{n} \bullet).
\end{equation}

Since the Killing spinor \eqref{eq:TypeDKS} is an element in $\SymSpin_{2,0}$, there are three possible actions within the algebra of symmetric spinors, namely by first contraction zero, one or two indices and then symmetrizing. Including a normalization, we define them as algebraic operators $\mathcal{K}^i:\SymSpin_{k,l} \to \SymSpin_{k-2i+2,l}, i=0,1,2$ via
\begin{subequations} \label{eq:Kprojectors}
\begin{align}
(\mathcal{K}^0 \varphi)_{\MindexA_{k+2} \MindexA'_l} \defeq{}&2\kappa_1^{-1} \kappa_{(A_1A_2}\varphi_{A_3\dots A_{k+2}) \MindexA'_l},\\
(\mathcal{K}^1 \varphi)_{\MindexA_k \MindexA'_l} \defeq{}&
\kappa_1^{-1}\kappa_{(A_1}{}^{F}\varphi_{A_2\dots A_k)F \MindexA'_l},\\
(\mathcal{K}^2 \varphi)_{\MindexA_{k-2} \MindexA'_l} \defeq{}&- \tfrac{1}{2}\kappa_1^{-1}\kappa^{CD} \varphi_{A_1\dots A_{k-2} CD \MindexA'_l}.
\end{align}
\end{subequations}
The commutators of the $\mathcal{K}$-operators and the extended fundamental operators are given in appendix~\ref{app:KopComm} and the complex conjugated operators act analogously on the primed indices.
\begin{definition}[Spin decomposition] \label{def:SpinDecomposition}
For any symmetric spinor $\varphi_{A_1 \dots A_{2s}}$ with integer $s$, define $s+1$ symmetric valence $2s$ spin-projectors $\mathcal{P}^{i}:\SymSpin_{2s,0} \to \SymSpin_{2s,0}, i = 0 \dots s$ solving
 \begin{align} \label{eq:SpinDecomposition}
  \varphi_{A_1 \dots A_{2s}} = \sum_{i=0}^{s} (\mathcal{P}^{i} \varphi)_{A_1 \dots A_{2s}},
 \end{align}
 with $(\mathcal{P}^{i} \varphi)_{A_1 \dots A_{2s}}$ depending only on the dyad components $\varphi_{s+i}$ and $\varphi_{s-i}$.
\end{definition}
These spin-projectors can be expressed in terms of the $\mathcal{K}$-operators and they are independent of the choice of principal frame, see \cite[Example 2.7]{2016arXiv160106084A}. The adjoints of the extended fundamental operators, the $\mathcal{K}$-operators and the spin-projectors are given in appendix~\ref{sec:adjoints}.

On a vacuum type D background in a principal dyad, the wave operator
\begin{align} \label{eq:DiagonalWaveOperator}
\sCurl_{(c-k,0)} \sCurlDagger_{(c,0)} 
\end{align}
is diagonal in $\SymSpin_{k,0}$ on its dyad components $\varphi_i, i = 0,\dots,k$ for any constant $c$. It follows that the spin-projectors commute with this diagonal wave operator.

\section{Spin-1} \label{sec:spin1}
In this section we consider second order symmetry operators for the source-free Maxwell equation
\begin{align} \label{eq:MaxwellEq}
(\sCurlDagger \phi)_{AA'}={}&0,
\end{align}
for the field strength $\phi_{AB} \in \SymSpin_{2,0}$. The field strength can locally be represented in terms of a real vector potential $\alpha_{AA'}$ via
\begin{align} \label{eq:MaxwellPotential}
 \phi_{AB} ={}&  (\sCurl \alpha)_{AB}.
\end{align}
The symmetry operators of this section naturally lead to complex vector potentials and we refer to remark \ref{rem:ComplexFields} for the general picture of symmetries for complex Maxwell fields.

\subsection{Maxwell symmetry operators}
We start by defining the operators 
\begin{subequations}
\begin{align}
\TMEEop&:\SymSpin_{1,1} \to \SymSpin_{1,1}, &
\TMEEop \defeq{}& \sCurlDagger \sCurl, \label{eq:Spin1Edef}\\
\TMETop&:\SymSpin_{1,1} \to \SymSpin_{2,0}, &
\TMETop \defeq{}& \sCurl. \label{eq:Spin1Tdef}
\end{align}
\end{subequations}
The first operator is the real Maxwell operator,
\begin{align} 
(\TMEEop \varphi)_{AA'}={}&\tfrac{1}{2} \nabla_{BB'}\nabla_{AA'}\varphi^{BB'}
 -  \tfrac{1}{2} \nabla_{BB'}\nabla^{BB'}\varphi_{AA'},
\end{align}
having vector potentials for source-free Maxwell solutions in its kernel and the second operator is the map \eqref{eq:MaxwellPotential} from a vector potential to its anti-self-dual field strength. Because of reality of $\TMEEop$ and properties of the adjoints of the fundamental operators given in appendix \ref{sec:adjoints}, we find
\begin{align} \label{eq:Spin1TMEEopAdjoints}
\TMEEop^\dagger  = \TMEEop^\star = \overline{\TMEEop} = \TMEEop.
\end{align}
The operator defined in \eqref{eq:Spin1Edef} differs from \cite{wald:1978PhRvL..41..203W} by a factor of $-2$. We also do not restrict the operator given in \eqref{eq:Spin1Tdef} to depend only on particular dyad components. To construct operator identities, define
 \begin{subequations}
  \begin{align} 
\TMESop&:\SymSpin_{1,1} \to \SymSpin_{2,0}, &
 \TMESop \defeq{}&\kappa_1^2 \mathcal{K}^1 \mathcal{K}^1 \sCurl_{(-2,0)}, \label{eq:Spin1Sdef}\\
\TMEOop&:\SymSpin_{2,0} \to \SymSpin_{2,0}, &
 \TMEOop \defeq{}&\kappa_1^2 \mathcal{K}^1 \mathcal{K}^1 \sCurl_{(-2,0)} \sCurlDagger, \label{eq:Spin1Odef} \\
\TSISop&:\SymSpin_{1,1} \to \SymSpin_{0,2}, &
\TSISop \defeq{}&\kappa_1 \bar{\kappa}_{1'} \overline{\mathcal{K}}^1 \sCurlDagger_{(-2,0)} \mathcal{K}^1, \label{eq:Spin1Shatdef}\\
\TSIOop&:\SymSpin_{2,0} \to \SymSpin_{0,2}, &
\TSIOop \defeq{}&\kappa_1 \bar{\kappa}_{1'} \overline{\mathcal{K}}^1 \sCurlDagger_{(-2,0)} \mathcal{K}^1 \sCurlDagger. \label{eq:Spin1Ohatdef}
 \end{align} 
 \end{subequations}
For any solution $\phi_{AB}$ to the source-free Maxwell equation \eqref{eq:MaxwellEq} it follows that $(\TMEOop \phi)_{AB}=0$ and $(\TSIOop \phi)_{A'B'}=0$. These are the covariant TME and TSI, respectively, see \eqref{eq:Spin1TMEGHPForm} and \eqref{eq:Spin1TSIGHPForm} for the GHP component form. We collect remarkable properties in
\begin{lemma} \label{lem:TMETSIadjoints}
 The TME operator is $\dagger$-self-adjoint and the TSI operator is $\star$-self-adjoint,
 \begin{align} \label{eq:Spin1TMETSIAdjoint}
 \TMEOop^\dagger = \TMEOop, &&&
 \TSIOop^\star = \TSIOop.
 \end{align}
 Furthermore the operators factorize into
 \begin{subequations} 
 \begin{align}
  \TMEOop ={}&- \TMEFop^\dagger \TMEFop, \label{eq:Spin1TMEOfactor}\\
  \TSIOop ={}&- \TSIFop^\star \TSIFop, \label{eq:Spin1TSIOfactor}
 \end{align}
 \end{subequations}
 with
\begin{subequations} 
\begin{align}
 \TMEFop&:\SymSpin_{2,0} \to \SymSpin_{1,1}, &
 \TMEFop \defeq{}& \sCurlDagger_{(1,0)} \mathcal{K}^1 (\kappa_1 \bullet),\\
 \TSIFop&:\SymSpin_{2,0} \to \SymSpin_{1,1}, &
 \TSIFop \defeq{}& \sCurlDagger \mathcal{K}^1 (\kappa_1 \bullet).
\end{align}
\end{subequations}
\end{lemma}
\begin{proof}
 The proof relies on rescalings of the form \eqref{eq:ExtendedFundSpinOpRescaling} and $\mathcal{K}$-operator commutators given in appendix~\ref{app:KopComm}. Commuting $\kappa_1\mathcal{K}^1$ from the left of the diagonal wave operator in \eqref{eq:Spin1Odef} to the right yields
\begin{align*}
\kappa_1^2 \mathcal{K}^1 \mathcal{K}^1 \sCurl_{(-2,0)} \sCurlDagger  ={}&\kappa_1 \mathcal{K}^1 \sCurl_{(-1,0)} \sCurlDagger_{(1,0)} \mathcal{K}^1 (\kappa_1 \bullet),
\end{align*}
from which \eqref{eq:Spin1TMEOfactor} follows. Commuting the $\mathcal{K}^1$ and the $\kappa_1$ in \eqref{eq:Spin1Ohatdef} to the right and arranging extended indices yields
\begin{align*}
\kappa_1 \bar{\kappa}_{1'} \overline{\mathcal{K}}^1 \sCurlDagger_{(-2,0)} \mathcal{K}^1 \sCurlDagger 
={}&\bar{\kappa}_{1'} \overline{\mathcal{K}}^1 \sCurlDagger_{(-1,0)} \sCurlDagger_{(1,0)} \mathcal{K}^1 (\kappa_1 \bullet) \\
={}&\bar{\kappa}_{1'} \overline{\mathcal{K}}^1 \sCurlDagger \sCurlDagger \mathcal{K}^1 (\kappa_1 \bullet),
\end{align*}  
from which \eqref{eq:Spin1TSIOfactor} follows. From these factorizations the self-adjointness \eqref{eq:Spin1TMETSIAdjoint} of TME and TSI follows from the property \eqref{eq:ProdAdjoint} of the adjoint of compositions.
\end{proof}
Lemma~\ref{lem:TMETSIadjoints} naturally leads to variational principles for the TME and the TSI. These, together with conservation laws, will be discussed in a separate paper. Here however, we use the self-adjointness to prove the main result of this section:

\begin{theorem} \label{thm:Spin1TMETSIOperatorIdentity}
Let $\phi_{AB}$ be a solution to the source-free Maxwell equation \eqref{eq:MaxwellEq} on a vacuum background of Petrov type D.
\begin{enumerate}
\item With the operators $\TMESop, \TMEEop, \TMEOop, \TMETop$ defined above we have the identity
\begin{align} \label{eq:Spin1TMEOperatorIdentity}
\TMESop \TMEEop ={}&\TMEOop \TMETop,
\end{align}
and the map
\begin{align}
\phi_{AB} \mapsto 
(\TMESop^\dagger \phi)_{AA'}={}&\bigl(\sCurlDagger_{(2,0)} \mathcal{K}^1 \mathcal{K}^1 (\kappa_1^2\phi)\bigr)_{AA'},\label{eq:PotentialSdaggerSpin1}
\end{align}
generates a new complex vector potential for a solution to the vacuum Maxwell equation.
\item With the operators $\TSISop, \TMEEop, \TSIOop, \TMETop$ defined above we have the identity
\begin{align} \label{eq:Spin1TSIOperatorIdentity}
\TSISop \TMEEop ={}&\TSIOop \TMETop,
\end{align}
and the map
\begin{align}
\phi_{AB} \mapsto 
(\TSISop^\star \phi)_{AA'}={}&\bigl(\overline{\mathcal{K}}^1 \sCurlDagger_{(0,2)} \mathcal{K}^1 (\kappa_1 \bar{\kappa}_{1'}\phi)\bigr)_{AA'},\label{eq:PotentialShatstarSpin1}
\end{align}
generates a new complex vector potential for a solution to the vacuum Maxwell equation.
\end{enumerate}
\end{theorem}
\begin{proof}
The operators can be viewed as different parts of third order operators,
\begin{align}
\rlap{$\overbrace{\phantom{\kappa_1^2\mathcal{K}^1 \mathcal{K}^1 \sCurl_{(-2,0)}}}^{\TMESop}$}
\rlap{$\underbrace{\phantom{\kappa_1^2 \mathcal{K}^1 \mathcal{K}^1  \sCurl_{(-2,0)}\sCurlDagger}}_{\TMEOop}$}
\kappa_1^2 \mathcal{K}^1 \mathcal{K}^1 \sCurl_{(-2,0)} 
\rlap{$\overbrace{\phantom{\sCurlDagger \sCurl}}^{\TMEEop}$} 
\sCurlDagger 
\rlap{$\underbrace{\phantom{\sCurl_{()}}}_{\TMETop}$}
\sCurl, 
&&
\rlap{$\overbrace{\phantom{\kappa_1 \bar{\kappa}_{1'}\overline{\mathcal{K}}^1 \sCurlDagger_{(-2,0)} \mathcal{K}^1}}^{\TSISop}$}
\rlap{$\underbrace{\phantom{\kappa_1 \bar{\kappa}_{1'}\overline{\mathcal{K}}^1 \sCurlDagger_{(-2,0)} \mathcal{K}^1 \sCurlDagger}}_{\TSIOop}$}
\kappa_1 \bar{\kappa}_{1'}\overline{\mathcal{K}}^1 \sCurlDagger_{(-2,0)} \mathcal{K}^1
\rlap{$\overbrace{\phantom{\sCurlDagger \sCurl}}^{\TMEEop}$} 
\sCurlDagger 
\rlap{$\underbrace{\phantom{\sCurlDagger_{()}}}_{\TMETop}$}
\sCurl,
\end{align}
so the identities \eqref{eq:Spin1TMEOperatorIdentity} and \eqref{eq:Spin1TSIOperatorIdentity} follow by construction. The rest follows from theorem~\ref{thm:AdjointMethod} by using \eqref{eq:Spin1TMEEopAdjoints} and \eqref{eq:Spin1TMETSIAdjoint}.
\end{proof}

\begin{remark} \label{rem:Spin1RelToClassification}
In \cite{ABB:symop:2014CQGra..31m5015A}, among others, second order symmetry operators for the source-free Maxwell equation \eqref{eq:MaxwellEq} are completely classified. In particular on a vacuum Petrov type D background the list contains, beside second Lie derivatives along isometries, one linear operator and one anti-linear operator. Both operators were presented in terms of complex vector potentials, $A_{AA'}$ and $B_{AA'}$, and a comparison reveals
\begin{subequations} 
\begin{align}
(\TMESop^\dagger \phi)_{AA'}
 ={}& 
 \bigl(\sCurlDagger_{(2,0)} \mathcal{K}^1 \mathcal{K}^1 (\kappa_1^2\phi)\bigr)_{AA'} \nonumber \\
 ={}&\bigl(\mathcal{K}^1 \sCurlDagger_{(2,0)} \mathcal{K}^1 (\kappa_1^2\phi)\bigr)_{AA'} \nonumber \\
 ={}& B_{AA'}, \\
(\TSISop^\star \phi)_{AA'}
={}&\bigl(\overline{\mathcal{K}}^1 \sCurlDagger_{(0,2)} \mathcal{K}^1 (\kappa_1 \bar{\kappa}_{1'}\phi)\bigr)_{AA'} \nonumber \\
={}&\tfrac{1}{3} (\sCurl \bar{\kappa})^{B}{}_{A'} \bigl(\mathcal{K}^1 (\kappa_1\phi)\bigr)_{AB}
 + \bar{\kappa}_{1'} \bigl(\overline{\mathcal{K}}^1 \sCurlDagger  \mathcal{K}^1 (\kappa_1\phi)\bigr)_{AA'} \nonumber \\
={}&A_{AA'}. 
\end{align}
\end{subequations} 
This shows that all irreducible, in the sense that they do not factor into first order symmetry operators, second order symmetry operators for the Maxwell equation \eqref{eq:MaxwellEq} on vacuum type D backgrounds follow from the adjoint operator argument.
\end{remark}
We collect further properties of the complex vector potentials in
\begin{corollary} \label{cor:Spin1PotentialDerivatives}
Let $\phi_{AB}$ be a solution to the source-free Maxwell equation \eqref{eq:MaxwellEq} on a vacuum background of Petrov type D.
\begin{enumerate} 
 \item The vector potential \eqref{eq:PotentialSdaggerSpin1} satisfies
 \begin{align}
 (\sDiv_{(2,0)} \TMESop^{\dagger}\phi)={}&0,\\
 (\TMETop \TMESop^{\dagger}\phi)_{AB}={}&(\TMEOop \phi)_{AB} = 0. \label{eq:Spin1CurlSdg}
 \end{align}
 The first equation can be interpreted as a generalized Lorenz gauge and the second one states that the anti-self-dual field strength of the vector potential vanishes on-shell. The self-dual field strength reads
 \begin{align} \label{eq:Spin1CurlDgSdg}
 \bar{\chi}_{A'B'}={}&(\overline{\TMETop} \TMESop^{\dagger}\phi)_{A'B'}.
 \end{align}
 This is the anti-linear symmetry operator of \cite{ABB:symop:2014CQGra..31m5015A}.  In particular it follows that taking the real part of the complex vector potential \eqref{eq:PotentialSdaggerSpin1} does not alter the (self-dual) field strength due to the complex conjugate of \eqref{eq:Spin1CurlSdg}. 
 \item The vector potential \eqref{eq:PotentialShatstarSpin1} satisfies
 \begin{align} \label{eq:Spin1CurlDgShatst}
 (\overline{\TMETop} \TSISop^{\star}\phi)_{A'B'}={}&(\TSIOop \phi)_{A'B'} = 0, 
 \end{align}
 which states that its self-dual field strength vanishes on-shell. The anti-self-dual field strength reads
 \begin{align} \label{eq:Spin1CurlShatst}
 \psi_{AB}={}&(\TMETop \TSISop^{\star}\phi)_{AB}.
 \end{align}
 This is the linear symmetry operator of \cite{ABB:symop:2014CQGra..31m5015A}.   Here again it follows that taking the real part of the complex vector potential \eqref{eq:PotentialSdaggerSpin1} does not alter the (anti-self-dual) field strength due to the complex conjugate of \eqref{eq:Spin1CurlDgShatst}.
\end{enumerate}
\end{corollary} 
In section \ref{sec:Spin2} we find that the complex potentials (metrics) in the linearized gravity case on a type D background do have both self-dual and anti-self-dual field strength (curvature) and hence taking the real part does have an effect in that case.
\begin{remark} \label{rem:ComplexFields} 
In general, for a complex vector potential $\alpha_{AA'}$, there is an anti-self-dual and a self-dual field strength
\begin{align}
\phi_{AB}={}&(\sCurl \alpha)_{AB}, &&&
\bar{\pi}_{A'B'}={}&(\sCurlDagger \alpha)_{A'B'},
\end{align}
solving the left and right Maxwell equations
\begin{align}
(\sCurlDagger \phi)_{AA'}={}&0, &&&
(\sCurl \bar{\pi})_{AA'}={}&0.
\end{align}
The equations are not coupled and the last equation can be read as the complex conjugate of $(\sCurlDagger \pi)_{AA'}=0$, so it is sufficient to analyze one of the equations. Also because $\TMEEop$ is a real operator, the above argument goes through for $\bar{\pi}_{A'B'}$ with the symmetry operators being complex conjugates of \eqref{eq:Spin1CurlDgSdg} and \eqref{eq:Spin1CurlShatst}.
Therefore, with constants $c_1,c_2,c_3,c_4$, the general irreducible symmetry operator reads
\begin{align}
{\left(\begin{array}{c}
\psi\\
\bar{\chi}
\end{array}\right)}={}&{\left(\begin{array}{cc}
c_1 \TMETop \TSISop^{\star} & c_2\TMETop \TMESop^{\star}\\
c_3 \overline{\TMETop} \TMESop^{\dagger} & c_4 \overline{\TMETop} \TSISop^{\dagger}
\end{array}\right)}{\left(\begin{array}{c}
\phi\\
\bar{\pi}
\end{array}\right)}.
\end{align}
\end{remark}

The operator $\mathcal{K}^1$ performs a sign-flip on one of the extreme components, for example in $\SymSpin_{2,0}$ the spinor $(\mathcal{K}^1\phi)_{AB}$ has GHP components
$ (\mathcal{K}^1\phi)_{0}=\phi_{0}, 
(\mathcal{K}^1\phi)_{1}=0,
(\mathcal{K}^1\phi)_{2}=- \phi_{2}$.
We note the following operator identity for $\TMESop^{\dagger}$ acting on the sign-flipped field,
\begin{align} \label{eq:Spin1SdgK1phi20}
(\TMESop^{\dagger}\mathcal{K}^1 \varphi)_{AA'}={}&(\kappa_1^2\mathcal{K}^1 \sCurlDagger  \varphi)_{AA'}
 + (\sTwist \mathcal{K}^2 \kappa_1^2\varphi)_{AA'}.
\end{align}
If $\varphi_{AB}$ is a source-free Maxwell field, it states that the left-hand side is a pure gauge vector potential because the right-hand side is a gradient (the first term vanishes for Maxwell fields). The analog for linearized gravity has recently been derived in \cite{2016arXiv160106084A}, see \eqref{eq:Spin2NiceIdentity} below, and will be used to construct a symmetry operator from the TSI in subsection~\ref{sec:Spin2TSI}. Next, we provide an interpretation of the two symmetry operators \eqref{eq:Spin1CurlDgSdg} and \eqref{eq:Spin1CurlShatst} in the sense of connecting it to the concepts of potentials and separability.

\subsection{Hertz potentials, Debye potentials and Teukolsky separability} \label{sec:Spin1Interpretation}

To discuss the anti-linear symmetry operator \eqref{eq:Spin1CurlDgSdg} we briefly recall the Hertz potential construction for spin-1 on a curved background along the lines of \cite[Section III]{1979PhRvD..19.1641K}. Let $\bar P_{A'B'} \in \SymSpin_{0,2}$ be a Hertz potential, i.e. solving the Hertz equation
\begin{align} \label{eq:Spin1HertzEquation}
(\sCurlDagger \sCurl \overline{P})_{A'B'}={}&-2 (\sCurlDagger \mathcal{G})_{A'B'},
\end{align}
with some (arbitrary) Nisbet gauge spinor $\mathcal{G}_{AA'}$ (on Minkowski space, setting $\mathcal{G}_{AA'}=0$, the Hertz equation is given by $\square\bar{P}_{A'B'}=0$). Then $\alpha_{AA'}$, generated via the Hertz map
\begin{align} \label{eq:Spin1HertzMap}
\alpha_{AA'} ={}& (\sCurl \overline{P})_{AA'} + 2 \mathcal{G}_{AA'},
\end{align} 
is a complex vector potential and its anti-self-dual field strength
\begin{align} \label{eq:Spin1HertzMapFieldStrength}
\chi_{AB}={}& (\sCurl \alpha)_{AB},
\end{align}
solves the Maxwell equation \eqref{eq:MaxwellEq} on any background as follows from a commutator. The self-dual field strength $(\sCurlDagger \alpha)_{A'B'}$ vanishes identically because of \eqref{eq:Spin1HertzEquation}. This implies that taking the real part of \eqref{eq:Spin1HertzMap} does not change the field strength \eqref{eq:Spin1HertzMapFieldStrength}. 

Restricting to a vacuum type D background and choosing
\begin{align} \label{eq:DiagonalGauge}
\mathcal{G}_{AA'}={}&\overline{P}_{A'}{}^{B'} \overline{U}_{B'A},
\end{align}
the Hertz equation \eqref{eq:Spin1HertzEquation} becomes
\begin{align} \label{eq:Spin1HertzEquationDiagonal}
(\sCurlDagger \sCurl_{(0,-2)} \overline{P})_{A'B'}={}&0,
\end{align}
which is a diagonal wave operator, c.f. the complex conjugate of \eqref{eq:DiagonalWaveOperator}. It is actually the complex conjugate of the TME \eqref{eq:Spin1Odef} if we set
\begin{align} \label{eq:Spin1MaxwellHertzPotential}
\overline{P}_{A'B'}= \bar{\kappa}_{1'}^2 (\overline{\mathcal{K}}^1 \overline{\mathcal{K}}^1 \bar{\phi})_{A'B'}.
\end{align} 
With this choice of $\overline{P}_{A'B'}$ and the gauge \eqref{eq:DiagonalGauge}, the field strength \eqref{eq:Spin1HertzMapFieldStrength} of the Hertz map is the complex conjugate of the anti-linear symmetry operator \eqref{eq:Spin1CurlDgSdg}. Note that only the gauge choice \eqref{eq:DiagonalGauge} converts the Hertz equation into the TME and only in this case can the Hertz map be interpreted as a symmetry operator. 

Because the Hertz equation in the form \eqref{eq:Spin1HertzEquationDiagonal} is diagonal, one can choose $\overline{P}_{A'B'}$ to have only one non-vanishing dyad component. This weighted scalar is called Debye potential and solves (up to rescaling by $\bar\kappa_{1'}$) either one of the scalar TMEs \eqref{eq:Spin1TMEGHPForm} in case it is an extreme component, or the Fackerell-Ipser equation $(\sDiv \sTwist_{(-2,0)} \mathcal{K}^2 \phi)=0$ in case it is the middle component. If one chooses the Hertz potential to be a (rescaled) Maxwell field then each of its three components, used as Debye potential, lead to the same new solution \eqref{eq:Spin1HertzMap}. To see this in terms of the symmetry operator it is important to make use of the freedom to modify the symmetry operator by terms vanishing due to the field equations. The situation is similar to the scalar wave operator $\square$ on Kerr and its Carter symmetry operator $Q = \nabla^a K_{ab}\nabla^b$ with $K_{ab}$ the Killing tensor. The linear combination $Q \pm \Sigma \square$, for a specific function $\Sigma$, is "purely radial" or "purely angular" and leads directly to separation of variables. We present the components of the anti-linear symmetry operator \eqref{eq:Spin1CurlDgSdg} in \eqref{eq:Spin1CurlDgSdgComps} in the appendix and here instead look at an alternative form in 
\begin{lemma} 
The anti-linear symmetry operator \eqref{eq:Spin1CurlDgSdg} can be represented in the form
\begin{align} \label{eq:Spin1CurlDgSdgModified}
\bar{\chi}_{A'B'}={}&(\overline{\TMETop} \TMESop^{\dagger}\phi)_{A'B'}
 \pm (\overline{\TMETop} \TMESop^{\dagger}\mathcal{K}^1 \phi)_{A'B'},
\end{align}
with the second term vanishing for solutions of the field equations, see \eqref{eq:Spin1SdgK1phi20}. For the plus sign the components are (see appendix~\ref{sec:GHPform})
\begin{subequations}  \label{eq:Spin12ndtypeSymGHP}
\begin{align} \label{eq:Spin12ndtypeSymGHPphi0}
\bar{\chi}_{0'}={}&2 \edt' \edt' (\kappa_1^2\phi_{0}), \nonumber \\
\bar{\chi}_{1'}={}&2 (\tho' \edt'
 + \bar{\tau}\tho' )(\kappa_1^2\phi_{0}), \\
\bar{\chi}_{2'}={}&2 \tho' \tho' (\kappa_1^2\phi_{0}),\nonumber
\end{align}
depending only on $\phi_0$ and for the minus sign they are
\begin{align} \label{eq:Spin12ndtypeSymGHPphi2}
\bar{\chi}_{0'}={}&2 \tho \tho (\kappa_1^2\phi_{2}),\nonumber \\
\bar{\chi}_{1'}={}&2 (\tho \edt
 + \bar{\tau}'\tho )(\kappa_1^2\phi_{2}),\\
\bar{\chi}_{2'}={}&2 \edt \edt (\kappa_1^2\phi_{2}),\nonumber
\end{align}
\end{subequations} 
depending only on $\phi_2$. 
\end{lemma}
Equations \eqref{eq:Spin12ndtypeSymGHP} are the Debye maps previously discussed e.g. in \cite{1979PhRvD..19.1641K}, \cite{wald:1978PhRvL..41..203W}, \cite{1985:TDCastillo}. The difference between \eqref{eq:Spin12ndtypeSymGHPphi0} and \eqref{eq:Spin12ndtypeSymGHPphi2} is the TSI \eqref{eq:Spin1TSIGHPForm} which leads us to interpret the anti-linear symmetry operator \eqref{eq:Spin1CurlDgSdg} as a covariant characterization of the Teukolsky-Starobinski constant, see \cite{kalnins:miller:1989} for an explicit proof in the Kerr case. From the above it also follows that the gradient term on the right-hand side of \eqref{eq:Spin1SdgK1phi20} maps between ingoing and outgoing radiation gauge.

The extreme components of the operator \eqref{eq:Spin1CurlDgSdg} can alternatively be made "purely angular" or "purely radial" by choosing another form as follows.
\begin{lemma} 
The anti-linear symmetry operator \eqref{eq:Spin1CurlDgSdg} can be represented in the form
\begin{align}
\bar{\chi}_{A'B'}={}&(\overline{\TMETop} \TMESop^{\dagger}\phi)_{A'B'}
 \pm (\overline{\mathcal{K}}^1 \kappa_1^2\sCurlDagger_{(-2,0)} \mathcal{K}^1 \sCurlDagger \phi)_{A'B'},
\end{align}
with the second term vanishing on solutions. For the plus sign the components are
\begin{subequations} \label{eq:Spin12ndtypeSymGHPangular} 
\begin{align}
\bar{\chi}_{0'}={}&2 \edt' \edt' (\kappa_1^2\phi_{0}),\\
\bar{\chi}_{1'}={}&(\tho \edt
 + \bar{\tau}'\tho )(\kappa_1^2\phi_{2})
 + (\tho' \edt'
 + \bar{\tau}\tho' )(\kappa_1^2\phi_{0}),\\
\bar{\chi}_{2'}={}&2 \edt \edt (\kappa_1^2\phi_{2}),
\end{align}
\end{subequations} 
while for the minus sign they are
\begin{subequations} \label{eq:Spin12ndtypeSymGHPradial} 
\begin{align}
\bar{\chi}_{0'}={}&2 \tho \tho (\kappa_1^2\phi_{2}),\\
\bar{\chi}_{1'}={}&(\tho \edt
 + \bar{\tau}'\tho )(\kappa_1^2\phi_{2})
 + (\tho' \edt'
 + \bar{\tau}\tho' )(\kappa_1^2\phi_{0}),\\
\bar{\chi}_{2'}={}&2 \tho' \tho' (\kappa_1^2\phi_{0}).
\end{align}
\end{subequations} 
\end{lemma} 
Note that the four representations \eqref{eq:Spin12ndtypeSymGHPphi0}, \eqref{eq:Spin12ndtypeSymGHPphi2}, \eqref{eq:Spin12ndtypeSymGHPangular}, \eqref{eq:Spin12ndtypeSymGHPradial} of the anti-linear symmetry operator lead to the same field $\bar{\chi}_{A'B'}$. This freedom in the representation may be important in the analysis of further properties of the symmetry operator.

Next, we consider the linear symmetry operator \eqref{eq:Spin1CurlShatst}. Its GHP components are given in \eqref{eq:Spin1CurlShatstComps} and here we consider again an alternative form with the additional term being, up to a multiplying function, the TME operator \eqref{eq:Spin1Odef}.
\begin{lemma} 
The linear symmetry operator \eqref{eq:Spin1CurlShatst} can be represented in the form
\begin{align} \label{eq:spin11sttypeSymMod}
\psi_{AB} ={}&(\TMETop \TSISop^{\star}\phi)_{AB}
 \pm \kappa_1 \bar{\kappa}_{1'} (\mathcal{K}^1 \mathcal{K}^1 \sCurl_{(-2,0)} \sCurlDagger \phi)_{AB},
\end{align}
with the second term vanishing on solutions. It leads to the "purely angular" extreme components
\begin{subequations} \label{eq:Spin11sttypeSymGHPAngular}
\begin{align}
\psi_{0}={}& 2 \edt \bigl(\kappa_1 \bar{\kappa}_{1'}(\edt'
 - 2 \tau')\bigr)\phi_{0} + \tfrac{1}{3} (\bar{\kappa}_{1'} \mathcal{L}_{\xi} -   \kappa_1 \mathcal{L}_{\bar{\xi}})\phi_{0},\\
\psi_{2}={}& 2 \edt' \bigl(\kappa_1 \bar{\kappa}_{1'}(\edt
 - 2 \tau)\bigr)\phi_{2} - \tfrac{1}{3} (\bar{\kappa}_{1'} \mathcal{L}_{\xi} - \kappa_1 \mathcal{L}_{\bar{\xi}})\phi_{2},
\end{align}
\end{subequations} 
for the plus sign and to the "purely radial" extreme components
\begin{subequations} \label{eq:Spin11sttypeSymGHPRadial}
\begin{align}
\psi_{0}={}& 2 \tho \bigl(\kappa_1 \bar{\kappa}_{1'}(\tho'
 - 2 \rho')\bigr)\phi_{0} 
 - \tfrac{1}{3} (\bar{\kappa}_{1'} \mathcal{L}_{\xi}
 + \kappa_1 \mathcal{L}_{\bar{\xi}})\phi_{0},\\
\psi_{2}={}& 2 \tho' \bigl(\kappa_1 \bar{\kappa}_{1'}(\tho
 - 2 \rho)\bigr)\phi_{2} + \tfrac{1}{3} (\bar{\kappa}_{1'} \mathcal{L}_{\xi}
 + \kappa_1 \mathcal{L}_{\bar{\xi}} )\phi_{2},
\end{align}
\end{subequations} 
for the minus sign. 
\end{lemma} 
See \eqref{eq:GHPLieXiphi20} for the Lie derivative of weighted scalars. The difference between \eqref{eq:Spin11sttypeSymGHPAngular} and \eqref{eq:Spin11sttypeSymGHPRadial} is the TME \eqref{eq:Spin1TMEGHPForm}, see also \cite[Section 5.4]{aksteiner:thesis}. This leads us to interpret the linear symmetry operator \eqref{eq:Spin1CurlShatst} as a covariant characterization of the TME separability (and therefore of the Teukolsky separation constant, see \cite{kalnins:miller:1989} for an explicit proof in the Kerr case). 
 
Summarized, we succeeded calculating both irreducible symmetry operators for Maxwell on vacuum type D backgrounds using the adjoint operator method and self-adjointness of the TME and TSI.

\section{Spin-2} \label{sec:Spin2}
Let $\delta g_{ABA'B'}=\delta g_{BAB'A'}$ be the spinorial form of a symmetric tensor field representing a linearized metric and define the irreducible parts
\begin{align}
G_{ABA'B'} \defeq{}&\delta g_{(AB)(A'B')}, & 
\slashed{G}_{} \defeq{}&\delta g^{C}{}_{C}{}^{C'}{}_{C'}.
\end{align} 
Note that $\delta g_{ABA'B'}$ is not a symmetric spinor, but $G_{ABA'B'}$ is. We use the covariant spinor variational operator $\vartheta$ developed in \cite{BaeVal15}. It is invariant under linearized tetrad rotations which allows us to do calculations covariantly. For relations to linearized dyad components (Newman-Penrose scalars), which involve the linearized tetrad, see  \cite[Remark 6]{BaeVal15}. We can express the variation of the three irreducible curvature spinors $\Psi_{ABCD}, \Phi_{ABA'B'}$ and $\Lambda$ on a vacuum background as differential operators acting on the linearized metric via
\begin{align}
\vartheta \Lambda={}&- \tfrac{1}{24} (\sDiv \sDiv G)
 + \tfrac{1}{32} (\sDiv \sTwist \slashed{G}),\\
\vartheta \Phi_{ABA'B'}={}&\tfrac{1}{2} G^{CD}{}_{A'B'} \Psi_{ABCD}
 + \tfrac{1}{2} (\sCurlDagger \sCurl G)_{ABA'B'}\nonumber\\
 & + \tfrac{1}{6} (\sTwist \sDiv G)_{ABA'B'} -  \tfrac{1}{8} (\sTwist \sTwist \slashed{G})_{ABA'B'},\\
\vartheta \Psi_{ABCD}={}&- \tfrac{1}{4} \slashed{G} \Psi_{ABCD}
 + \tfrac{1}{2} (\sCurl \sCurl G)_{ABCD}. \label{eq:VarSPsi}
\end{align}

It is convenient to introduce a modification of the linearized Weyl spinor $\vartheta\Psi_{ABCD}$,
\begin{align} \label{eq:phidef} 
\phi_{ABCD} \defeq{}& \vartheta \Psi_{ABCD} + \tfrac{1}{4} \slashed{G}_{} \Psi_{ABCD}.
\end{align}
In a type D principal frame this modification only affects the middle component. A variation of the Einstein spinor on a vacuum background without sources leads to the spinorial form of the linearized Einstein equation
\begin{align} \label{eq:LinearizedEinsteinEquation}
\vartheta\mathcal{E} 
\defeq -2\vartheta\Phi_{ABA'B'} -6 \epsilon_{AB}\bar{\epsilon}_{A'B'} \vartheta\Lambda  = 0.
\end{align} 
Multiplied by a factor of $-2$ the operator reads
\begin{align} \label{eq:LinearizedEinsteinSpinor}
-2 \vartheta\mathcal{E}  ={}&- \nabla_{AA'}\nabla_{BB'}\delta g^{C}{}_{C}{}^{C'}{}_{C'}
 + \nabla_{CC'}\nabla_{AA'}\delta g_{B}{}^{C}{}_{B'}{}^{C'}\nonumber\\
& + \nabla_{CC'}\nabla_{BB'}\delta g_{A}{}^{C}{}_{A'}{}^{C'}
 -  \nabla_{CC'}\nabla^{CC'}\delta g_{ABA'B'}\nonumber\\
& -  \epsilon_{AB} \bar\epsilon_{A'B'} \nabla_{DD'}\nabla_{CC'}\delta g^{CDC'D'}\nonumber\\
& + \epsilon_{AB} \bar\epsilon_{A'B'} \nabla_{DD'}\nabla^{DD'}\delta g^{C}{}_{C}{}^{C'}{}_{C'}.
\end{align}
A computation shows that this operator is self-adjoint. On a Petrov type D background, the irreducible components of \eqref{eq:LinearizedEinsteinSpinor} lead to the matrix equation
\begin{align*}
{\left(\begin{array}{c}
\vartheta \Phi\\
3\vartheta \Lambda
\end{array}\right)}={}&{\left(\begin{array}{cc}
 \TMEEop & - \tfrac{1}{8}\sTwist \sTwist \\
- \tfrac{1}{8}\sDiv \sDiv & \tfrac{3}{32}\sDiv \sTwist 
\end{array}\right)}{\left(\begin{array}{c}
G\\
\slashed{G}
\end{array}\right)},
\end{align*}
with
\begin{align}\label{eq:Spin2Edef}
\TMEEop&:\SymSpin_{2,2} \to \SymSpin_{2,2}, \nonumber \\
\TMEEop &\defeq{} \tfrac{1}{2} \sCurlDagger \sCurl 
 + \tfrac{1}{6} \sTwist \sDiv 
 + \tfrac{1}{2} \Psi_{2} (\mathcal{K}^1\mathcal{K}^1 
 -   \mathcal{K}^0 \mathcal{K}^2).
\end{align}
This operator is real and self-adjoint,
\begin{align} \label{eq:Spin2TMEEopAdjoints}
 \TMEEop^\dagger = \TMEEop^\star = \overline{\TMEEop} = \TMEEop,
\end{align}
as follows from appendix~\ref{sec:adjoints}. To define the $\TMETop$-operator, we consider the map from linearized metric to its Weyl-curvature \eqref{eq:VarSPsi}. Because the trace term only contributes to the middle component on a type D background in a principal frame, and we are interested in the extreme components only, we define
\begin{align}
\TMETop&:\SymSpin_{2,2} \to \SymSpin_{4,0}, &
\TMETop \defeq{}&\tfrac{1}{2} \sCurl \sCurl, \label{eq:TMETop}
\end{align}
so that \eqref{eq:phidef} is equivalent to $(\TMETop G)_{ABCD} = \phi_{ABCD}$. Because the equations are considerably more complicated for linearized gravity than in the spin-1 case, we investigate the TME and TSI separately in the following two subsections and interpret the resulting symmetry operators in subsection~\ref{sec:Spin2Interpretation}.

\subsection{The TME and a fourth order symmetry operator} \label{sec:Spin2TME}
In this section we derive an operator identity based on the TME, analogous to \eqref{eq:Spin1TMEOperatorIdentity} in the spin-1 case. Because \eqref{eq:Spin2Edef} is the linearized trace-free Ricci spinor of a trace-free metric, $(\TMEEop G)_{ABA'B'} = (\vartheta \Phi[G,0])_{ABA'B'}$, we can use the spin-2 TME with sources, derived in \cite[eq.(3.13)]{2016arXiv160106084A}, given by
\begin{align} \label{eq:Spin2TME}
\kappa_1^4 ( \mathcal{P}^{2} \mathcal{K}^1 \sCurl_{(-4,0)} \sCurl
\vartheta \Phi)_{ABCD}
={}& (\sCurl \sCurlDagger_{(4,0)} \mathcal{K}^1 \mathcal{P}^{2} \kappa_1^4\vartheta\Psi)_{ABCD} \nonumber \\
&+ 3 \Psi_{2} (\mathcal{K}^1 \mathcal{P}^{2} \kappa_1^4\vartheta\Psi)_{ABCD}. \hspace{5pt}
\end{align}
Motivated by \eqref{eq:Spin2TME}, define the operators
\begin{align} 
\TMESop&:\SymSpin_{2,2} \to \SymSpin_{4,0}, &
\TMESop \defeq{}&\kappa_1^4\mathcal{P}^{2} \sCurl_{(-4,0)} \sCurl, \label{eq:Spin2SopDef}\\
\TMEOop&:\SymSpin_{4,0} \to \SymSpin_{4,0}, &
\TMEOop \defeq{}&\sCurl \sCurlDagger_{(4,0)} \kappa_1^4\mathcal{P}^{2} 
 + 3 \Psi_{2} \kappa_1^4 \mathcal{P}^{2}. \label{eq:Spin2OopDef}
\end{align}
Equation~\eqref{eq:Spin2TME} could have been written without the $\mathcal{K}^1$ operators, see the proof of theorem~\ref{thm:Spin2TMEOperatorIdentity} below for details. 

The the modified linearized curvature $\phi_{ABCD}$ of any source-free solution of the linearized Einstein equation \eqref{eq:LinearizedEinsteinEquation} on a vacuum type D background solves the TME 
\begin{align} \label{eq:Spin2TMEOop}
(\TMEOop \phi)_{ABCD} = 0.
\end{align}
See \eqref{eq:Spin2TMEGHPForm} for the GHP component form. Remarkable properties of this operator are summarized in
\begin{lemma}
The TME operator is $\dagger$-self-adjoint,
\begin{align} \label{eq:TMEOopSelfadj}
\TMEOop^\dagger ={}&\TMEOop.
\end{align}
Furthermore it factorizes, up to a potential term, into
\begin{align}
\TMEOop ={}& \TMEFop^\dagger \TMEFop + 3 \Psi_{2} \kappa_1^4 \mathcal{P}^{2},
\end{align}
with
\begin{align}
\TMEFop&:\SymSpin_{4,0} \to \SymSpin_{3,1}, &
\TMEFop \defeq{}&\sCurlDagger_{(2,0)} \mathcal{P}^{2}( \kappa_1^2 \bullet).
\end{align}
\end{lemma}
\begin{proof}
Because \eqref{eq:Spin2OopDef} is diagonal, we can apply another (idempotent) spin-2 projector without altering the result and commute $\kappa_1^2$ out to get
\begin{align*}
\TMEOop
={} \kappa_1^2 \mathcal{P}^{2} \sCurl_{(-2,0)} \sCurlDagger_{(2,0)} (\mathcal{P}^{2} \kappa_1^2  \bullet) + 3 \Psi_{2} \kappa_1^4 \mathcal{P}^{2} \\
={}\TMEFop^\dagger \TMEFop + 3 \Psi_{2} \kappa_1^4 \mathcal{P}^{2},
\end{align*}
From this, the $\dagger$-self-adjointness of $\TMEOop$ follows.
\end{proof}

\begin{theorem} \label{thm:Spin2TMEOperatorIdentity}
With the operators $\TMESop, \TMEEop, \TMEOop, \TMETop$ defined above we have the identities
\begin{subequations}
\begin{align} 
\TMESop \TMEEop={}&\TMEOop \TMETop, \label{eq:Spin2TMEOperatorIdentitya} \\
\TMESop \sTwist \sTwist ={}& 0. \label{eq:Spin2TMEOperatorIdentityb}
\end{align}
\end{subequations} 
If $G_{ABA'B'}$ is the trace-free part of a solution to the source-free vacuum linearized Einstein equation \eqref{eq:LinearizedEinsteinEquation} with modified curvature $\phi_{ABCD}=(\TMETop G)_{ABCD}$,  then the map
\begin{align} \label{eq:Spin2LinMetricSdagger}
\phi_{ABCD} \mapsto 
(\TMESop^\dagger \phi)_{ABA'B'}={}&(\sCurlDagger \sCurlDagger_{(4,0)} \mathcal{P}^{2} \kappa_1^4\phi)_{ABA'B'},
\end{align}
generates a new complex solution to the source-free vacuum linearized Einstein equation.
\end{theorem} 
\begin{proof}
Applying another $\mathcal{K}^1$ to \eqref{eq:Spin2TME}, using $\mathcal{K}^1 \mathcal{K}^1 \mathcal{P}^2 = \mathcal{P}^2$  and identifying the operators \eqref{eq:Spin2OopDef} leads to identity \eqref{eq:Spin2TMEOperatorIdentitya} for the trace-free part of the linearized metric. The right-hand side of \eqref{eq:Spin2TME} does not depend on $\slashed{G}$ and therefore $(\TMESop \vartheta \Phi[0,\slashed{G}])_{A'B'C'D'}=0$, where $(\vartheta \Phi[0,\slashed{G}])_{ABA'B'}=(\sTwist \sTwist \slashed{G})_{ABA'B'}$. This proves the identity \eqref{eq:Spin2TMEOperatorIdentityb}.

From theorem~\ref{thm:AdjointMethod} together with the self-adjointness of $ \TMEEop$ and $\TMEOop$  given in \eqref{eq:Spin2TMEEopAdjoints} and \eqref{eq:TMEOopSelfadj}, respectively, it follows that \eqref{eq:Spin2LinMetricSdagger} maps into the kernel of $ \TMEEop$. The adjoint of \eqref{eq:Spin2TMEOperatorIdentityb} yields
\begin{align} \label{eq:DivDivSdagger}
\sDiv \sDiv \TMESop^\dagger \phi={}&0,
\end{align}
which ensures that \eqref{eq:Spin2LinMetricSdagger} has vanishing Ricci scalar curvature. Hence it generates new solutions to linearized gravity from solutions to the TME \eqref{eq:Spin2TMEOop}.
\end{proof}
For later reference, we define
\begin{align} \label{eq:hmetric}
h_{ABA'B'} \defeq (\TMESop^\dagger \phi)_{ABA'B'}
\end{align}
for the new complex metric generated from TME solutions via \eqref{eq:Spin2LinMetricSdagger}. Analogous to corollary \ref{cor:Spin1PotentialDerivatives} for spin-1, we collect the curvature of the new solution in
\begin{lemma}
The complex linearized metric \eqref{eq:hmetric} has self-dual and anti-self-dual curvature
\begin{subequations} \label{eq:Spin2TMECurvature}
\begin{align}
\bar{\chi}_{A'B'C'D'}={}&(\overline{\TMETop} h)_{A'B'C'D'} \nonumber \\
={}&\tfrac{1}{2} (\sCurlDagger \sCurlDagger \sCurlDagger \sCurlDagger_{(4,0)} \mathcal{P}^{2}\kappa_1^4\phi)_{A'B'C'D'},\label{eq:Spin2TMEsdCurvature}\\
\psi_{ABCD}={}&(\TMETop h)_{ABCD} \nonumber \\
={}&\tfrac{1}{2} (\sCurl \sCurl \sCurlDagger \sCurlDagger_{(4,0)} \mathcal{P}^{2} \kappa_1^4\phi)_{ABCD} \nonumber \\
={}&\tfrac{1}{2} \Psi_{2} \kappa_1^3 (\mathcal{L}_{\xi}\mathcal{K}^1 \mathcal{P}^{2} \phi)_{ABCD}. \label{eq:Spin2TMEasdCurvature}
\end{align}
\end{subequations} 
The last equality holds on-shell.
\end{lemma}
\begin{proof}
The first equation is just an expansion of the operators and \eqref{eq:Spin2TMEasdCurvature} follows from
\begin{align} \label{eq:Spin2TMEasdCurvatureId}
\psi_{ABCD}
={}&\tfrac{1}{2} (\sCurl \sCurlDagger \TMEOop \phi)_{ABCD} - \tfrac{1}{2} \Psi_{2} (\TMEOop \phi)_{ABCD} \nonumber \\
& + \tfrac{1}{2} \Psi_{2} \kappa_1^3 (\mathcal{L}_{\xi}\mathcal{K}^1 \mathcal{P}^{2} \phi)_{ABCD},
\end{align}
which is an operator identity when starting with the first equality in \eqref{eq:Spin2TMEasdCurvature}. This operator identity follows from a lengthy calculation involving the commutator of $\sCurl \sCurlDagger $.
\end{proof}
Note that the anti-self-dual curvature \eqref{eq:Spin2TMEasdCurvature} reduces to first order and \eqref{eq:Spin2TMEasdCurvatureId} is the analog of \eqref{eq:Spin1CurlSdg}. Restricting to the real or imaginary part of \eqref{eq:Spin2LinMetricSdagger} leads to a mixture of the two curvatures. An interpretation of the symmetry operator is given in subsection~\ref{sec:Spin2Interpretation}.

\subsection{The TSI and a sixth order symmetry operator} \label{sec:Spin2TSI}
In this section we derive an operator identity based on the TSI, analogous to \eqref{eq:Spin1TSIOperatorIdentity} in the spin-1 case. The covariant form of the spin-2 TSI can be deduced from the identity, derived in \cite[eq.(4.19)]{2016arXiv160106084A},
\begin{align} \label{eq:Spin2NiceIdentity}
(\TMESop^{\dagger}\mathcal{K}^1 \phi)_{ABA'B'}
={}& (\sTwist \mathcal{A})_{ABA'B'} 
+ \tfrac{1}{2} \Psi_{2} \kappa_1^3 (\mathcal{L}_{\xi}G)_{ABA'B'} \nonumber \\
&+ (\TSINop \vartheta \Phi)_{ABA'B'},
\end{align}
where the $\TMESop^{\dagger}$ operator is given in \eqref{eq:Spin2LinMetricSdagger} and
\begin{align}
\TSINop \defeq{}&- \kappa_1^4 \mathcal{K}^0 \mathcal{K}^2 \sTwist_{(-7,0)} \mathcal{K}^2 \sCurl  - 3 \Psi_{2} \kappa_1^4 \mathcal{K}^1 \nonumber\\
& + \kappa_1^4 \mathcal{K}^1 \sCurlDagger_{(-4,0)} \sCurl 
 + \tfrac{4}{3} \kappa_1^4 \sTwist_{(-7,0)} \mathcal{K}^2 \sCurl, \label{eq:TSINop} \\
\mathcal{A}_{AA'} \defeq{}&- \tfrac{1}{2} \Psi_{2} \kappa_1^3 \xi^{BB'} (\mathcal{K}^0 \mathcal{K}^2 G)_{ABA'B'}  + \tfrac{2}{3} \kappa_1^3 \xi^{B}{}_{A'} (\mathcal{K}^1 \mathcal{K}^2 \phi)_{AB}\nonumber\\
& + (\mathcal{K}^1 \sTwist  \mathcal{K}^2 \mathcal{K}^2 (\kappa_1^4\phi))_{AA'}
 -  \tfrac{1}{4} \Psi_{2} \kappa_1^4 (\mathcal{K}^1 \sTwist_{(2,0)} \slashed{G})_{AA'}\nonumber\\
&+ \tfrac{2}{3} \kappa_{1}{}^4 (\mathcal{K}^2\sCurl\vartheta \Phi)_{AA'}.
\end{align}
The right (and therefore also the left) hand side of \eqref{eq:Spin2NiceIdentity}  can be shown to be a complex, trace-free solution to the source-free linearized Einstein equation if $(G_{ABA'B'}, \slashed{G})$ is, see \cite[Cor.4.3]{2016arXiv160106084A}. The analogue of \eqref{eq:Spin2NiceIdentity} in the spin-1 case is \eqref{eq:Spin1SdgK1phi20} and there it is a pure gauge vector potential. Here the last term on the right-hand side of \eqref{eq:Spin2NiceIdentity} is the trace-free part of a linearized diffeomorphism but the first term is not. 

By applying $\tfrac{1}{2}\sCurlDagger \sCurlDagger $, or equivalently the $\overline{\TMETop}$ operator, to \eqref{eq:Spin2NiceIdentity} we get the corresponding self-dual curvature. It turns out to be convenient to apply $\overline{\mathcal{P}}^{2}$ to pick out the extreme components and a $\overline{\mathcal{K}}^1$ operator to flip the sign on one of them.
This combination of operators on the Lie derivative term just gives a Lie derivative of the complex conjugated curvature. The $(\sTwist \mathcal{A})_{ABA'B'}$ term is the trace-free part of a linearized diffeomorphism, and will therefore not contribute to the gauge independent extreme components of the curvature. As we will see in Theorem~\ref{thm:Spin2TSIOperatorIdentity}, the remaining terms can be compactly expressed in terms of the operators
\begin{subequations}\label{eq:TSIOperatorDefs}
\begin{align}
\TSIOop&: \SymSpin_{4,0} \to \SymSpin_{0,4}, &
\TSIOop \defeq{}&2 \kappa_1^{-1} \bar{\kappa}_{1'}^3 \overline{\mathcal{P}}^{2} \overline{\mathcal{K}}^1 \overline{\TMETop} \TMESop^{\dagger}\mathcal{K}^1,\label{eq:TSIOperator}\\
\TSISop&: \SymSpin_{2,2} \to \SymSpin_{0,4}, &
\TSISop \defeq{}&2 \kappa_1^{-1} \bar{\kappa}_{1'}^3 \overline{\mathcal{P}}^{2} \overline{\mathcal{K}}^1 \overline{\TMETop} \TSINop,\\
\TSIUop&: \SymSpin_{0,4} \to \SymSpin_{0,4}, &
\TSIUop \defeq{}&\Psi_{2} \kappa_1^2 \bar{\kappa}_{1'}^3 \overline{\mathcal{P}}^{2} \overline{\mathcal{K}}^1 \mathcal{L}_{\xi}.\label{eq:TSILOperator}
\end{align}
\end{subequations}
In particular, the modified curvature $\phi_{ABCD}=(\TMETop G)_{ABCD}$ of any source-free solution to the linearized Einstein equation \eqref{eq:LinearizedEinsteinSpinor} on a vacuum type D background solves the TSI
\begin{align}\label{eq:TSIEqOperatorForm}
(\TSIOop \phi)_{A'B'C'D'} - (\TSIUop \overline{\phi})_{A'B'C'D'} = 0.
\end{align}
See \eqref{eq:Spin2TSIOpGHPForm} for the GHP component form. Remarkable properties of these operators are summarized in
\begin{lemma}
The operator \eqref{eq:TSIOperator} is $\star$-self-adjoint,
\begin{align}\label{eq:TSIOopSelfadj}
\TSIOop^\star={}&\TSIOop,
\end{align}
and the Lie derivative term \eqref{eq:TSILOperator} is $\dagger$-self-adjoint,
\begin{align}\label{eq:TSIUopSelfadj}
\TSIUop^{\dagger}={}&\TSIUop.
\end{align}
Furthermore \eqref{eq:TSIOperator} factorizes into
\begin{align}
\TSIOop ={}&- \TSIFop^\star \TSIFop,
\end{align}
with
\begin{align}
\TSIFop&: \SymSpin_{4,0} \to \SymSpin_{2,2}, &
\TSIFop \defeq{}&\sCurlDagger \sCurlDagger  \mathcal{K}^1 \kappa_1^3\mathcal{P}^{2}.
\end{align}
\end{lemma}
\begin{proof}
Using an extended index identity, we get
\begin{align*}
\TSIOop ={}&\overline{\mathcal{P}}^{2} \bar{\kappa}_{1'}^3\overline{\mathcal{K}}^1 \sCurlDagger \sCurlDagger \sCurlDagger \sCurlDagger \mathcal{K}^1 \kappa_1^3\mathcal{P}^{2}
={}- \TSIFop^{\star}\TSIFop.
\end{align*}
From this, the $\star$-self-adjointness of $\TSIOop$ is manifest. \eqref{eq:TSIUopSelfadj} follows directly from the adjoints given in appendix~\ref{sec:adjoints}.
\end{proof}

\begin{theorem}\label{thm:Spin2TSIOperatorIdentity}
With the operators $\TSISop, \TMEEop, \TSIOop, \TMETop$ defined above we have the identities 
\begin{subequations} \label{eq:Spin2TSIOperatorIndentity}
\begin{align} 
 \TSISop \TMEEop ={}& \TSIOop \TMETop - \TSIUop \overline{\TMETop}, \label{eq:SEOTSpin2TSI}\\
 \TSISop \sTwist \sTwist ={}& 0. \label{eq:STwistTwistSpin2TSI}
 \end{align}
 \end{subequations} 
If $G_{ABA'B'}$ is the trace-free part of a solution to the source-free vacuum linearized Einstein equation \eqref{eq:LinearizedEinsteinEquation} with modified curvature $\phi_{ABCD}=(\TMETop G)_{ABCD}$,  then the real part of the map
\begin{align} \label{eq:Spin2Shatstar}
\phi_{ABCD} \mapsto 
(\TSISop^{\star}\phi)_{ABA'B'}={}&-2 (\TSINop^{\star} \TMETop^{\dagger}\mathcal{K}^1 \mathcal{P}^{2} \kappa_1^3\bar{\kappa}_{1'}^{-1}\phi)_{ABA'B'}\quad 
\end{align}
generates a new solution to the source-free vacuum linearized Einstein equation. 
\end{theorem}
For convenience we state the adjoints of \eqref{eq:TSINop} and \eqref{eq:TMETop} explicitly,
\begin{subequations}
\begin{align}
\TSINop^{\star}={}&
 \tfrac{1}{3} \bar{\kappa}_{1'}^4 \sCurl_{(0,-4)} \overline{\mathcal{K}}^0 \sDiv_{(0,3)} 
 - \bar{\kappa}_{1'}^4 \sCurl_{(0,-4)} \sCurlDagger \overline{\mathcal{K}}^1  \nonumber \\
 &-  \tfrac{1}{4} \bar{\kappa}_{1'}^4 \sCurl_{(0,-4)} \overline{\mathcal{K}}^0 \sDiv_{(0,3)} \overline{\mathcal{K}}^0 \overline{\mathcal{K}}^2  
 + 3 \bar\Psi_{2} \bar{\kappa}_{1'}^4 \overline{\mathcal{K}}^1,\\
\TMETop^{\dagger} ={}&\tfrac{1}{2} \sCurlDagger \sCurlDagger.
\end{align}
\end{subequations}
\begin{proof} 
Applying the operator $\overline{\mathcal{P}}^{2} \overline{\mathcal{K}}^1 \overline{\TMETop}$ to the identity \eqref{eq:Spin2NiceIdentity} and identifying the different pieces with the operators \eqref{eq:TSIOperatorDefs} gives the relation \eqref{eq:SEOTSpin2TSI} after finding that the $(\sTwist \mathcal{A})_{ABA'B'}$ term can be seen as the trace-free part of a linearized diffeomorphism, which will not contribute to the gauge invariant extreme components of the curvature. The only terms in \eqref{eq:Spin2NiceIdentity} that depends on $\slashed{G}$ are $(\sTwist \mathcal{A})_{ABA'B'}$ and $(\TSINop \vartheta \Phi)_{ABA'B'}$, but as we have already concluded, the $(\sTwist \mathcal{A})_{ABA'B'}$ term does not contribute to the extreme components of the curvature. Therefore $(\TSISop \vartheta \Phi[0,\slashed{G}])_{A'B'C'D'}=0$, where $(\vartheta \Phi[0,\slashed{G}])_{ABA'B'}=(\sTwist \sTwist \slashed{G})_{ABA'B'}$. This gives the identity \eqref{eq:STwistTwistSpin2TSI}.

The $\star$-adjoint of \eqref{eq:SEOTSpin2TSI}, using the self-adjointness of the operators $\TSIOop$ \eqref{eq:TSIOopSelfadj}, $\TSIUop$ \eqref{eq:TSIUopSelfadj} and $\TMEEop$ \eqref{eq:Spin2TMEEopAdjoints}, yields
\begin{align}
(\TMEEop \TSISop^\star \varphi)_{ABA'B'} = (\TMETop^\star \TSIOop \varphi)_{ABA'B'} - (\overline{\TMETop}^\star \overline{\TSIUop} \varphi)_{ABA'B'}.
\end{align}
We want to use a solution of the TSI \eqref{eq:TSIEqOperatorForm} to generate new solutions, but the right-hand side of the last equation is not of that form. However, taking the real part of the "new metric", we find
\begin{align}
(\TMEEop (\TSISop^\star \phi + \TSISop{}^\dagger \bar\phi))_{ABA'B'} ={}& 
(\TMETop^\star (\TSIOop \phi - \TSIUop \bar\phi))_{ABA'B'} \nonumber \\
&+ (\TMETop^\dagger (\overline{\TSIOop} \bar\phi - \overline{\TSIUop} \phi ))_{ABA'B'}.
\end{align}
Now the two terms on the right-hand side contain the right-hand side of \eqref{eq:SEOTSpin2TSI} and its complex conjugate. Therefore solutions of the TSI \eqref{eq:TSIEqOperatorForm} generate metrics in the kernel of the $\TMEEop$ operator. 
The adjoint of \eqref{eq:STwistTwistSpin2TSI} gives
\begin{align}
\sDiv \sDiv \TSISop^\star \phi={}&0.
\end{align}
As the $\sDiv \sDiv$ operator is real, we see that also the linearized scalar curvature of the new metric vanishes. Hence, the mapping $\phi_{ABCD}\mapsto \Re(\TSISop^\star \phi)_{ABA'B'}$ to the real part generates new solutions to linearized gravity from solutions to the TSI equation \eqref{eq:TSIEqOperatorForm}.
\end{proof}

For later reference, we define
\begin{align} \label{eq:kmetric}
k_{ABA'B'} \defeq \Re(\TSISop^\star \phi)_{ABA'B'}
\end{align}
for the new real metric generated from TSI solutions via \eqref{eq:Spin2Shatstar}. Since the metric is real, its self-dual curvature is the complex conjugate of the anti-self-dual curvature and hence it is sufficient to calculate the latter via
\begin{align}\label{eq:kMetricCurvature}
\psi_{ABCD} ={}& (\TMETop k)_{ABCD} \nonumber \\
={}& \tfrac{1}{2} (\TMETop \TSISop^\star \phi)_{ABCD} + \tfrac{1}{2} (\TMETop \TSISop^\dagger \bar\phi)_{ABCD}.
\end{align}
The first part cannot be reduced by the field equations, but at least the extreme components of the second part can be reduced due to the following identity,
\begin{align}\label{eq:CurlDgCurlDgShatStarIdentity}
2(\overline{\mathcal{P}}^{2} \overline{\TMETop}  \TSISop^{\star}\phi)_{A'B'C'D'}
={}&(\overline{\mathcal{P}}^{2} \sCurlDagger \sCurl \TSIOop \phi)_{A'B'C'D'} \nonumber \\
&- \bar\Psi_{2} (\overline{\mathcal{P}}^{2} \TSIOop \phi)_{A'B'C'D'}.
\end{align}
Next we provide an interpretation along the lines of the spin-1 case.

\subsection{Hertz potentials, Debye potentials and Teukolsky separability} \label{sec:Spin2Interpretation}

To discuss the complex metric \eqref{eq:hmetric}, we briefly recall the Hertz potential construction for linearized gravity on a vacuum type D background similar to \cite[Section V]{1979PhRvD..19.1641K}.  Let $\bar P_{A'B'C'D'}$ be a Hertz potential with vanishing non-extreme components w.r.t. a principal dyad, i.e. a symmetric spinor solving the Hertz equation
\begin{align} \label{eq:Spin2HertzEquation}
(\sCurlDagger \sCurl_{(0,4)} \bar{P})_{A'B'C'D'} + 3 \bar\Psi_{2}\bar{P}_{A'B'C'D'} ={}&0.
\end{align}
Then the complex, symmetric spinor $H_{ABA'B'}$ generated via the Hertz map,
\begin{align} \label{eq:Spin2HertzMap}
H_{ABA'B'}={}&(\sCurl \sCurl_{(0,4)} \bar{P})_{ABA'B'},
\end{align}
solves the linearized Einstein equation \eqref{eq:LinearizedEinsteinEquation}.
\begin{remark}
In \cite[Section V]{1979PhRvD..19.1641K} the more general Hertz equation
\begin{align}
(\sCurlDagger \sCurl \bar{P})_{A'B'C'D'} + 3 \bar{P}_{A'B'C'D'} \bar\Psi_{2} ={}&(\sCurlDagger \mathcal{G})_{A'B'C'D'},
\end{align}
and Hertz map
\begin{align}
H_{ABA'B'}={}& (\sCurl \sCurl \bar{P})_{ABA'B'} - (\sCurl \mathcal{G})_{ABA'B'},
\end{align}
with a "gauge" spinor $\mathcal{G}_{AA'B'C'}$, were proposed. We checked that the linearized Ricci scalar vanishes, $ \vartheta \Lambda[H] = 0 $, but for the trace-free Ricci spinor components, we find e.g.
\begin{align}
(\TMEEop H)_{00'}={}&-6 \bar\Psi_{2} (\edt
 + \bar{\tau}')(\tfrac{1}{4} \mathcal{G}_{00'}
 + \bar{P}_{0'} \bar{\tau}'),
\end{align}
which fixes $\mathcal{G}_{00'}$. Therefore $\mathcal{G}_{AA'B'C'}$ is not a freely specifiable gauge field on a curved background. For the case $ \mathcal{G}_{AA'B'C'} = -4 \bar{P}_{A'B'C'D'} \bar{U}^{D'}{}_{A}$, which has components 
$\mathcal{G}_{00'}=-4 \bar{P}_{0'} \bar{\tau}',
\mathcal{G}_{01'}=0,
\mathcal{G}_{02'}=0,$
$\mathcal{G}_{03'}=4 \bar{P}_{4'} \bar{\rho},
\mathcal{G}_{10'}=-4 \bar{P}_{0'} \bar{\rho}',
\mathcal{G}_{11'}=0,
\mathcal{G}_{12'}=0,
\mathcal{G}_{13'}=4 \bar{P}_{4'} \bar{\tau}$
we end up with \eqref{eq:Spin2HertzEquation}, \eqref{eq:Spin2HertzMap}. This choice was made in the course of the proof in \cite[Section V]{1979PhRvD..19.1641K}.
\end{remark}
The operator in \eqref{eq:Spin2HertzEquation} is diagonal, c.f. the complex conjugate of \eqref{eq:DiagonalWaveOperator}.  It is actually the complex conjugate of the TME operator \eqref{eq:Spin2OopDef} if we set
\begin{align} \label{eq:Spin2LinGravHertzPotential}
\overline{P}_{A'B'C'D'}= \bar{\kappa}_{1'}^4 (\overline{\mathcal{P}}^2 \bar{\phi})_{A'B'C'D'}.
\end{align} 
With this choice of $\overline{P}_{A'B'C'D'}$, the Hertz map \eqref{eq:Spin2HertzMap} is the complex conjugate of the symmetry operator \eqref{eq:Spin2LinMetricSdagger} viewed as a map from metric to metric, $H_{ABA'B'} = \bar h_{ABA'B'}$. This shows that the spin-2 Hertz potential formalism on vacuum type D backgrounds can be understood as the symmetry operator \eqref{eq:Spin2LinMetricSdagger}.

Because \eqref{eq:Spin2HertzEquation} is diagonal, one can choose the Hertz potential $\overline{P}_{A'B'C'D'}$ to have only one non-vanishing extreme dyad component. This weighted scalar is called Debye potential and solves (up to rescaling by $\bar\kappa_{1'}$) by construction one of the complex conjugated scalar TMEs \eqref{eq:Spin2TMEGHPForm}. 

If we use the linearized curvature as a Hertz potential via \eqref{eq:Spin2LinGravHertzPotential}, the extreme components used as Debye potentials generate different new solutions to the linearized Einstein equation. However, the difference is not very complicated and we derive it explicitly. We do this in two steps. First a modification of the symmetry operator analogous to \eqref{eq:Spin1CurlDgSdgModified} is made, but this modification is not pure gauge on a curved background. Then in the second step we add the correction term to have a pure gauge modification and to show that both extreme curvature scalars generate the same new solution to linearized gravity.

A modification of the symmetry operator \eqref{eq:Spin2LinMetricSdagger} of the form (the second term is a solution on its own due to \eqref{eq:Spin2NiceIdentity}, but it is not pure gauge)
\begin{align} \label{eq:Spin2LinMetrichpm}
\underline{h}^{\pm}_{ABA'B'} =  
(\TMESop^\dagger \phi)_{ABA'B'} \pm (\TMESop^{\dagger}\mathcal{K}^1 \phi)_{ABA'B'}
\end{align} 
is again a complex solution to linearized gravity and it depends only on one of the extreme curvature scalars ($\phi_0$ for $+$ and $\phi_4$ for $-$). These are the Debye potential maps given by the Hertz map \eqref{eq:Spin2HertzMap} with one of the extreme components set to zero. However, we will see that $\underline{h}^+ \neq \underline{h}^-$ and contrary to the spin-1 case the difference is not pure gauge. For completeness we present the components of the self-dual curvature $\bar{\underline{\chi}}^\pm_{A'B'C'D'}=(\overline{\TMETop} \underline{h}^\pm)_{A'B'C'D'}$ and anti-self-dual curvature $\underline{\psi}^\pm_{ABCD}=(\TMETop \underline{h}^\pm)_{ABCD}$ of \eqref{eq:Spin2LinMetrichpm} for the plus sign in \eqref{eq:Spin2LinMetrichpCurvGHP} and for the minus sign in \eqref{eq:Spin2LinMetrichmCurvGHP} in the appendix. 

For the second step we note that a pure gauge metric can be constructed from \eqref{eq:Spin2NiceIdentity} and in the source-free case we have analogous to \eqref{eq:Spin1SdgK1phi20}
\begin{align} 
(\TMESop^{\dagger}\mathcal{K}^1 \phi)_{ABA'B'}- \tfrac{1}{2} \Psi_{2} \kappa_1^3 (\mathcal{L}_{\xi}G)_{ABA'B'}={}& (\sTwist \mathcal{A})_{ABA'B'}, \label{eq:Spin1SdgK1phi20a}\\
- \tfrac{1}{2} \Psi_{2} \kappa_1^3 (\mathcal{L}_{\xi}\slashed{G})={}& (\sDiv \mathcal{A}),
\end{align}
where the second equation is given in \cite[eq.(4.13)]{2016arXiv160106084A}. The right-hand side are the trace-free and trace parts of a linearized diffeomorphism (one can discuss the real and imaginary parts separately and deal with real diffeomorphisms), so we can add/subtract this to/from the linearized metric \eqref{eq:hmetric} without changing the actual perturbation,
\begin{subequations} 
\begin{align} 
h_{ABA'B'} ={}&  
(\TMESop^\dagger \phi)_{ABA'B'} \pm \bigl((\TMESop^{\dagger}\mathcal{K}^1 \phi)_{ABA'B'} \nonumber \\
& \hspace{55pt}- \tfrac{1}{2} \Psi_{2} \kappa_1^3 (\mathcal{L}_{\xi}G)_{ABA'B'}\bigr) \label{eq:Spin2LinMetricSdaggerb} \\
={}&  
\underline{h}^{\pm}_{ABA'B'} \mp \tfrac{1}{2} \Psi_{2} \kappa_1^3 (\mathcal{L}_{\xi}G)_{ABA'B'}, \label{eq:hTohpm} \\
\slashed{h} ={}& \mp \tfrac{1}{2} \Psi_{2} \kappa_1^3 (\mathcal{L}_{\xi}\slashed{G}).
\end{align}
\end{subequations} 
Here we understand $(h_{ABA'B'},\slashed{h})$ as a representative of gauge equivalent metrics, and therefore the equalities here are up to gauge.
It follows that the difference between $\underline{h}^+$ and $\underline{h}^-$ is a Lie derivative and a gauge transformation. Also note that the gauge transformation introduces a trace to the linearized metric. Now, using \eqref{eq:hTohpm}, \eqref{eq:Spin2LinMetrichpCurvGHP} and \eqref{eq:Spin2LinMetrichmCurvGHP} we get the different forms of the curvature of the metric $h_{ABA'B'}$. 

\begin{lemma}
The extreme components of the self-dual and anti-self-dual curvatures of the complex metric $h_{ABA'B'}$ are
\begin{align}
\begin{aligned}\label{eq:hcurvatureGHPplus}
\bar{\chi}_{0'}={}&  \edt' \edt' \edt' \edt' (\kappa_1^4\phi_{0}) - \tfrac{1}{2} \Psi_{2} \kappa_1^3 \mathcal{L}_{\xi}\bar{\phi}_{0'}, \\
\bar{\chi}_{4'}={}& \tho' \tho' \tho' \tho' (\kappa_1^4\phi_{0}) - \tfrac{1}{2} \Psi_{2} \kappa_1^3 \mathcal{L}_{\xi}\bar{\phi}_{4'}, \\
\psi_{0}={}&\tfrac{1}{2} \Psi_{2} \kappa_1^3 \mathcal{L}_{\xi}\phi_{0}, \\
\psi_{4}={}&- \tfrac{1}{2} \Psi_{2} \kappa_1^3 \mathcal{L}_{\xi}\phi_{4},
\end{aligned}
\end{align} 
for the plus case and
\begin{align} 
\begin{aligned}\label{eq:hcurvatureGHPminus}
\bar{\chi}_{0'}={}&  \tho \tho \tho \tho (\kappa_1^4\phi_{4}) + \tfrac{1}{2} \Psi_{2} \kappa_1^3 \mathcal{L}_{\xi}\bar{\phi}_{0'}, \\
\bar{\chi}_{4'}={}& \edt \edt \edt \edt (\kappa_1^4\phi_{4}) + \tfrac{1}{2} \Psi_{2} \kappa_1^3 \mathcal{L}_{\xi}\bar{\phi}_{4'}, \\
\psi_{0}={}&\tfrac{1}{2} \Psi_{2} \kappa_1^3 \mathcal{L}_{\xi}\phi_{0}, \\
\psi_{4}={}&- \tfrac{1}{2} \Psi_{2} \kappa_1^3 \mathcal{L}_{\xi}\phi_{4},
\end{aligned}
\end{align} 
for the minus case. As the corresponding metrics only differ by a linearized diffeomorphism, these extreme curvature components are the same. 
\end{lemma}
Restricting to the real or imaginary part of the metric \eqref{eq:Spin2LinMetricSdaggerb} leads to linear combinations $\frac{1}{2}(\overline{\chi}_{n} + \overline{\psi}_{n})$ or $\frac{1}{2i}(\overline{\chi}_{n} - \overline{\psi}_{n})$ for the self-dual curvature. The Lie derivative terms makes both \eqref{eq:hcurvatureGHPplus} and \eqref{eq:hcurvatureGHPminus} dependent of $\phi_0$ and $\phi_4$. 
\begin{remark}
If we study the imaginary part of $h_{ABA'B'}$ we get for the Kerr case $\Im( h)_{ABA'B'}=\Im(\underline{h}^{\pm})_{ABA'B'}$ because $G_{ABA'B'}$ and $\xi_{AA'}$ are real and $\Psi_{2} \kappa_1^3$ is a real constant in that case. Hence, it reduces down to the Debye potential case with self-dual curvature $\frac{1}{2i}(\overline{\underline{\chi}}^{\pm}_{n} - \overline{\underline{\psi}}^{\pm}_{n})$ given by \eqref{eq:Spin2LinMetrichpCurvGHP} and \eqref{eq:Spin2LinMetrichmCurvGHP}.
\end{remark}
 
An alternative point of view can be obtained by noting that $\overline{\mathcal{P}}^{2}\overline{\TMETop}$ on \eqref{eq:Spin1SdgK1phi20a} gives
\begin{subequations}
\begin{align}
(\overline{\mathcal{P}}^{2} \overline{\TMETop} \sTwist \mathcal{A})_{A'B'C'D'}
={}&(\overline{\mathcal{P}}^{2} \overline{\TMETop}  \TMESop^{\dagger}\mathcal{K}^1 \phi)_{A'B'C'D'} \nonumber \\
&\, - \tfrac{1}{2} \Psi_{2} \kappa_1^3 (\overline{\mathcal{P}}^{2} \mathcal{L}_{\xi}\bar{\phi})_{A'B'C'D'} \\
={}&\tfrac{1}{2} \kappa_1 \bar{\kappa}_{1'}^{-3} \bigl((\overline{\mathcal{K}}^1 \TSIOop  \phi)_{A'B'C'D'} \nonumber\\
&\hspace{35pt} - (\overline{\mathcal{K}}^1 \TSIUop \bar{\phi})_{A'B'C'D'}\bigr).
\end{align}
\end{subequations}
The left-hand side vanishes because the extreme curvature components are gauge invariant and the right-hand side vanishes because of the TSI. Summarized, the extreme components of $\phi_{ABCD}$ can be used as a Hertz potential via \eqref{eq:Spin2LinGravHertzPotential}, or each one of them as a Debye potential. Our analysis shows that the difference between these three possibilities are the TSI and Lie derivatives.

Similar to \eqref{eq:Spin12ndtypeSymGHPangular} and \eqref{eq:Spin12ndtypeSymGHPradial} for the spin-1 case, we can write the extreme components of the curvature of \eqref{eq:hmetric} in "purely angular" or "purely radial" form by adding/subtracting the TSI in a different way.
\begin{lemma} 
The extreme components of the self-dual curvature \eqref{eq:Spin2TMEsdCurvature} can be represented in the form
\begin{align}
\bar{\chi}_{A'B'C'D'}={}&(\overline{\TMETop} \TMESop^{\dagger}\phi)_{A'B'C'D'} \pm \tfrac{1}{2} \kappa_1 \bar{\kappa}_{1'}^{-3} \bigl((\TSIOop \phi)_{A'B'C'D'} \nonumber \\
& \hspace{100pt}-  (\TSIUop \bar{\phi})_{A'B'C'D'}\bigr).
\end{align}
with the second term vanishing on solutions. The components read
\begin{subequations}\label{eq:Spin2TMESymasdCurvMod3}
\begin{align}
\bar{\chi}_{0'}={}&\edt' \edt' \edt' \edt' (\kappa_1^4\phi_{0}) - \tfrac{1}{2} \Psi_{2} \kappa_1^3 \mathcal{L}_{\xi}\bar{\phi}_{0'},\\
\bar{\chi}_{4'}={}&\edt \edt \edt \edt (\kappa_1^4\phi_{4}) + \tfrac{1}{2} \Psi_{2} \kappa_1^3 \mathcal{L}_{\xi}\bar{\phi}_{4'}.
\end{align}
\end{subequations}
for the plus sign and
\begin{subequations}\label{eq:Spin2TMESymasdCurvMod4}
\begin{align}
\bar{\chi}_{0'}={}& \tho \tho \tho \tho (\kappa_1^4\phi_{4}) + \tfrac{1}{2} \Psi_{2} \kappa_1^3 \mathcal{L}_{\xi}\bar{\phi}_{0'},\\
\bar{\chi}_{4'}={}& \tho' \tho' \tho' \tho' (\kappa_1^4\phi_{0}) - \tfrac{1}{2} \Psi_{2} \kappa_1^3 \mathcal{L}_{\xi}\bar{\phi}_{4'}.
\end{align}
\end{subequations}
for the minus sign. 
\end{lemma} 
These alternative forms of the symmetry operator play an important role for its invertibility.

Finally, we consider the sixth order symmetry operator \eqref{eq:Spin2Shatstar}. Recall that  the analogous (i.e. from the TSI operator identity) symmetry operator in the spin-1 case encoded TME separability. This cannot work here due to the mismatched number of derivatives of the symmetry operator and the TME (however, see remark~\ref{rem:Spin2TMESymop}) but we find the following
\begin{lemma} \label{lem:Spin21stkindComps}
On a Schwarzschild background, the extreme components of the curvature \eqref{eq:kMetricCurvature},
\begin{align}
\psi_{ABCD} ={}& \tfrac{1}{2} (\TMETop \TSISop^\star \phi)_{ABCD} + \tfrac{1}{2} (\TMETop \TSISop^\dagger \bar\phi)_{ABCD},\label{eq:kMetricCurvature2}
\end{align}
after simplification due to field equations take the form
\begin{align}
\begin{aligned}\label{eq:SixthOrderSymOpSch}
\psi_{0}={}& \kappa_1^6 (\edt \edt'
 + 2 \Psi_{2}
 + 2 \rho \rho')(\edt \edt'
 + \Psi_{2}
 + \rho \rho')\edt \edt' \phi_{0},\\
 \psi_{4}={}& \kappa_1^6 (\edt' \edt
 + 2 \Psi_{2}
 + 2 \rho \rho')(\edt' \edt
 + \Psi_{2}
 + \rho \rho')\edt' \edt \phi_{4}.
\end{aligned}
\end{align}
In a principal frame we have $\Psi_2 = - \frac{M}{r^3}, \rho\rho' = \frac{2M-r}{2r^3}$ with $M$ the mass of the black hole and $r$ the areal radius coordinate. Since $\kappa_1 \propto \Psi_2^{-1/3}$ we choose $\kappa_1=-r/3$ so that $ \kappa_1^2 \Psi_{2}  + \kappa_1^2 \rho \rho'= - 1/18$ and \eqref{eq:SixthOrderSymOpSch} is almost the spin-weighted spherical Laplacian to the power three (the difference is a constant shift in the eigenvalues).
\end{lemma}
This can be seen through direct component calculations and the following argument. As we assume that $\phi_{ABCD}$ satisfies both TME and TSI, we can use the identity \eqref{eq:CurlDgCurlDgShatStarIdentity} to reduce the order of the second term in \eqref{eq:kMetricCurvature2} due to the TSI. The first term however, will remain sixth order, but we can use the TME to eliminate the $\tho\tho'$ derivatives after commutations. The lower order terms from the second term in \eqref{eq:kMetricCurvature2} cancel with the lower order terms from the first term so that we finally end up with \eqref{eq:SixthOrderSymOpSch}.

\begin{remark} \label{rem:Spin2TMESymop}
It should be noted that on a generalized Kerr-NUT spacetime (real $\xi_{AA'}$), the operator (similar to \eqref{eq:spin11sttypeSymMod} in the spin-1 case) defined by
\begin{align}\label{eq:SecondOrderTMESymop}
\mathrm{S}&: \SymSpin_{4,0} \to \SymSpin_{4,0}, \nonumber \\
\mathrm{S} &\defeq{} \kappa_1\mathcal{P}^{2}\mathcal{K}^1\sCurl_{(-1,1)} \bar{\kappa}_{1'}\overline{\mathcal{K}}^1 \sCurlDagger_{(-3,0)} \mathcal{P}^{2} - \kappa_1\mathcal{K}^1 \mathcal{P}^{2} \mathcal{L}_{\xi},
\end{align}
is a second order symmetry operator for the TME. In GHP form, the non-vanishing components are (c.f. Theorem 5.4.1 in \cite{aksteiner:thesis})
\begin{subequations}
\begin{align}
(\mathrm{S}\phi)_{0}={}&2 \kappa_1 \bar{\kappa}_{1'} (\edt
 -  \tau
 -  \bar{\tau}')(\edt'
 - 4 \tau')\phi_{0} \nonumber \\
 &+ (\bar{\kappa}_{1'}- \kappa_1) \mathcal{L}_{\xi}\phi_{0}
  -  \kappa_1^{-3} \bar{\kappa}_{1'} (\TMEOop \phi)_{0},\\
(\mathrm{S} \phi)_{4}={}&2 \kappa_1 \bar{\kappa}_{1'} (\edt'
 -  \bar{\tau}
 -  \tau')(\edt
 - 4 \tau)\phi_{4} \nonumber \\
 &+ (\kappa_1
 -  \bar{\kappa}_{1'}) \mathcal{L}_{\xi}\phi_{4}
  -  \kappa_1^{-3} \bar{\kappa}_{1'} (\TMEOop \phi)_ {4}.
\end{align}
\end{subequations}
If we assume that $\phi_{0}$ and $\phi_{4}$ satisfies the TME, then on a Schwarzschild spacetime we get $(\mathrm{S}\phi)_{0}=2 \kappa_1^2 \edt\edt'\phi_{0}$ and $(\mathrm{S} \phi)_{4}=2 \kappa_1^2 \edt' \edt\phi_{4}$. Therefore, \eqref{eq:SixthOrderSymOpSch} can be written in terms of  the $\mathrm{S}$ operator, which gives a relation to the TME separation constants. It is an open question if the sixth order operator can be factored also on the Kerr spacetime. Even though \eqref{eq:SecondOrderTMESymop} is a symmetry operator for the TME, it can not be interpreted as a symmetry operator for linearized gravity, but the sixth order operator comes from a linearized metric, and can therefore be cast into a form mapping linearized metrics to linearized metrics.
\end{remark}

\section{Conclusions} \label{sec:conclusion}

In this paper we have shown that for the Maxwell equations and linearized gravity on vacuum spacetimes of Petrov type D, the covariant TME and TSI equations can separately be cast into self-adjoint form. This was used to construct a symmetry operator in each of the cases which was then related to various concepts like Hertz and Debye potentials or TME separability. Moreover, the self-adjointness naturally leads to variational principles for the TME and TSI. A recent application of the corresponding canonical energy for the spin-2 TME in the Schwarzschild geometry for the discussion of linear stability by Prabhu and Wald is given in \cite{Prabhu:2018jvy}. We expect the symmetry operators to play an important role in the general study of decay estimates for spin-1 and spin-2 similar to the scalar wave equation case in \cite{andersson:blue:0908.2265}. The modifications of the operators with terms vanishing on-shell may open up the possibility to invert certain potential maps and lead to a generalization of the decay results of \cite{2014CMaPh.331..755A} to a curved background. Finally, the self-adjoint TSI for linearized gravity leads to a new conservation law which we plan to discuss in a separate paper.

\appendix

\section{Adjoints} \label{sec:adjoints}
In this section we collect the $\dagger$- and $\star$-adjoints of the algebraic and differential operators introduced in section~\ref{sec:Preliminaries}. First of all, for a general composition of operators $\mathbf{A}$ and $\mathbf{B}$ we have
\begin{align} \label{eq:ProdAdjoint}
 (\mathbf{A}\mathbf{B})^\dagger = \mathbf{B}^\dagger \mathbf{A}^\dagger, &&&
 (\mathbf{A}\mathbf{B})^\star = \mathbf{B}^\star \mathbf{A}^\star.
\end{align}
For general constants $n,m$, the adjoints of the extended fundamental spinor operators \eqref{eq:ExtendedFundSpinOp} are given by
\begin{subequations}
\begin{align}
(\sDiv_{(n,m)})^\dagger ={}& - \sTwist_{(-n,-m)}, &&&
(\sCurl_{(n,m)})^\dagger ={}& \sCurlDagger_{(-n,-m)}, \\
(\sTwist_{(n,m)})^\dagger ={}& - \sDiv_{(-n,-m)}, &&&
(\sCurlDagger_{(n,m)})^\dagger ={}& \sCurl_{(-n,-m)},\\
(\sDiv_{(n,m)})^\star ={}& - \sTwist_{(-m,-n)}, &&&
(\sCurl_{(n,m)})^\star ={}& \sCurl_{(-m,-n)}, \\
(\sTwist_{(n,m)})^\star ={}& - \sDiv_{(-m,-n)}, &&&
(\sCurlDagger_{(n,m)})^\star ={}& \sCurlDagger_{(-m,-n)}.
\end{align}
\end{subequations} 
The adjoints of the algebraic $\mathcal{K}$-operators \eqref{eq:Kprojectors} are
\begin{subequations}
\begin{align}
\mathcal{K}^0{}^{\dagger}={}&-4 \mathcal{K}^2, &
\mathcal{K}^1{}^{\dagger}={}&- \mathcal{K}^1, &
\mathcal{K}^2{}^{\dagger}={}&- \tfrac{1}{4} \mathcal{K}^0, \\
\overline{\mathcal{K}}^0{}^{\dagger}={}&-4 \overline{\mathcal{K}}^2, &
\overline{\mathcal{K}}^1{}^{\dagger}={}&- \overline{\mathcal{K}}^1, &
\overline{\mathcal{K}}^2{}^{\dagger}={}&- \tfrac{1}{4} \overline{\mathcal{K}}^0, \\
\mathcal{K}^0{}^{\star}={}&-4 \overline{\mathcal{K}}^2, &
\mathcal{K}^1{}^{\star}={}&- \overline{\mathcal{K}}^1, &
\mathcal{K}^2{}^{\star}={}&- \tfrac{1}{4} \overline{\mathcal{K}}^0, \\
\overline{\mathcal{K}}^0{}^{\star}={}&-4 \mathcal{K}^2, &
\overline{\mathcal{K}}^1{}^{\star}={}&- \mathcal{K}^1, &
\overline{\mathcal{K}}^2{}^{\star}={}&- \tfrac{1}{4} \mathcal{K}^0.
\end{align}
\end{subequations}
Multiplication by a scalar, e.g. $\kappa_1$ or $\Psi_2$, is a self-adjoint operation. For the projection operators \eqref{eq:SpinDecomposition} we find
\begin{align}
(\mathcal{P}^{s})^{\dagger}={}&\mathcal{P}^{s}, &
(\mathcal{P}^{s})^{\star}={}&\overline{\mathcal{P}}^{s}.
\end{align}
The Lie derivative is skew-adjoint,
\begin{align}
(\mathcal{L}_{\xi})^{\dagger} ={}&- \mathcal{L}_{\xi}.
\end{align}

\section{\texorpdfstring{$\mathcal{K}$}{K}-operator commutators} \label{app:KopComm}

To avoid clutter in the notation we present the commutators as operators which can be applied to arbitrary elements of $\SymSpin_{(k,l)}$, where $k$ and $l$ are large enough so that the combination of operators on the left hand side is properly defined. $n,m$ used below are arbitrary constants. The proof is straightforward but tedious.Examples can be found in \cite[Lemma 2.9]{2016arXiv160106084A}. Complex conjugating these identities gives commutators for the $\overline{\mathcal{K}}$ operators.
\begin{lemma}
Commuting $\mathcal{K}$-operator outside the extended fundamental spinor operators on $\SymSpin_{k,l}$ yields
\begin{subequations}
\begin{align}
\sDiv_{(n,m)}\mathcal{K}^2
={}&\mathcal{K}^2 \sDiv_{(n+1,m)},\\
\sCurlDagger_{(n,m)}\mathcal{K}^2
={}&\mathcal{K}^2 \sCurlDagger_{(n+1,m)},\\
\sCurl_{(n,m)}\mathcal{K}^2
={}&\mathcal{K}^2 \sCurl_{(n-1,m)}
-\tfrac{1}{k+1}\mathcal{K}^1 \sDiv_{(n+k,m)},\\
\sTwist_{(n,m)} \mathcal{K}^2
={}&\mathcal{K}^2 \sTwist_{(n-1,m)}
-\tfrac{1}{k+1}\mathcal{K}^1 \sCurlDagger_{(n+k,m)},
\end{align}
\begin{align}
\sDiv_{(n,m)}\mathcal{K}^1
={}&\tfrac{(k-1)(k+2)}{k(k+1)}\mathcal{K}^1 \sDiv_{(n,m)}
+\tfrac{2}{k}\mathcal{K}^2 \sCurl_{(n-k-1,m)},\\
\sCurlDagger_{(n,m)}\mathcal{K}^1
={}&\tfrac{(k-1)(k+2)}{k(k+1)}\mathcal{K}^1 \sCurlDagger_{(n,m)}
+\tfrac{2}{k}\mathcal{K}^2 \sTwist_{(n-k-1,m)},\\
\sCurl_{(n,m)}\mathcal{K}^1
={}&\mathcal{K}^1 \sCurl_{(n,m)}
+\tfrac{1}{2(k+1)}\mathcal{K}^0 \sDiv_{(n+k+1,m)},\\
\sTwist_{(n,m)}\mathcal{K}^1
={}&\mathcal{K}^1 \sTwist_{(n,m)}
+\tfrac{1}{2(k+1)}\mathcal{K}^0 \sCurlDagger_{(n+k+1,m)},
\end{align}
\begin{align}
\sDiv_{(n,m)}\mathcal{K}^0
={}&\tfrac{k(k+3)}{(k+1)(k+2)}\mathcal{K}^0 \sDiv_{(n-1,m)}
-\tfrac{4}{k+2}\mathcal{K}^1 \sCurl_{(n-k-2,m)},\\
\sCurlDagger_{(n,m)}\mathcal{K}^0
={}&\tfrac{k(k+3)}{(k+1)(k+2)}\mathcal{K}^0 \sCurlDagger_{(n-1,m)}
-\tfrac{4}{k+2}\mathcal{K}^1 \sTwist_{(n-k-2,m)},\\
\sCurl_{(n,m)}\mathcal{K}^0
={}&\mathcal{K}^0 \sCurl_{(n+1,m)},\\
\sTwist_{(n,m)}\mathcal{K}^0
={}&\mathcal{K}^0 \sTwist_{(n+1,m)}.
\end{align}
\end{subequations}
\end{lemma}
\begin{corollary}
Commuting $\mathcal{K}$-operator inside the extended fundamental spinor operators on $\SymSpin_{k,l}$ yields
\begin{subequations}
\begin{align}
\mathcal{K}^2 \sDiv_{(n,m)}
={}&\sDiv_{(n-1,m)}\mathcal{K}^2,\\
\mathcal{K}^2 \sCurlDagger_{(n,m)}
={}&\sCurlDagger_{(n-1,m)}\mathcal{K}^2,\\
\mathcal{K}^2 \sCurl_{(n,m)}
={}&\tfrac{(k-1)(k+2)}{k(k+1)}\sCurl_{(n+1,m)}\mathcal{K}^2
+\tfrac{1}{k+1}\sDiv_{(n+k+1,m)}\mathcal{K}^1,\\
\mathcal{K}^2 \sTwist_{(n,m)}
={}
&\tfrac{(k-1)(k+2)}{k(k+1)}\sTwist_{(n+1,m)}\mathcal{K}^2
+\tfrac{1}{k+1}\sCurlDagger_{(n+k+1,m)}\mathcal{K}^1,
\end{align}
\begin{align}
\mathcal{K}^1 \sDiv_{(n,m)}
={}
&\sDiv_{(n,m)}\mathcal{K}^1
-\tfrac{2}{k}\sCurl_{(n-k,m)}\mathcal{K}^2,\\
\mathcal{K}^1 \sCurlDagger_{(n,m)}
={}
&\sCurlDagger_{(n,m)}\mathcal{K}^1
-\tfrac{2}{k}\sTwist_{(n-k,m)}\mathcal{K}^2, \\
\mathcal{K}^1 \sCurl_{(n,m)}={}
&\tfrac{k(k+3)}{(k+1)(k+2)}\sCurl_{(n,m)}\mathcal{K}^1
-\tfrac{1}{2(k+1)}\sDiv_{(n+k+2,m)}\mathcal{K}^0,\\
\mathcal{K}^1 \sTwist_{(n,m)}={}
&\tfrac{k(k+3)}{(k+1)(k+2)}\sTwist_{(n,m)}\mathcal{K}^1
-\tfrac{1}{2(k+1)}\sCurlDagger_{(n+k+2,m)}\mathcal{K}^0,
\end{align}
\begin{align}
\mathcal{K}^0 \sDiv_{(n,m)}={}
&\sDiv_{(n+1,m)}\mathcal{K}^0
+\tfrac{4}{k+2}\sCurl_{(n-k-1,m)}\mathcal{K}^1,\\
\mathcal{K}^0\sCurlDagger_{(n,m)}={}
&\sCurlDagger_{(n+1,m)}\mathcal{K}^0
+\tfrac{4}{k+2}\sTwist_{(n-k-1,m)}\mathcal{K}^1,\\
\mathcal{K}^0\sCurl_{(n,m)}
={}&\sCurl_{(n-1,m)}\mathcal{K}^0,\\
\mathcal{K}^0 \sTwist_{(n,m)}
={}&\sTwist_{(n-1,m)}\mathcal{K}^0.
\end{align}
\end{subequations}
\end{corollary}
\begin{lemma}
Commuting two $\mathcal{K}$-operators on $\SymSpin_{k,l}$ yields
\begin{subequations}
\begin{align}
\mathcal{K}^2 \mathcal{K}^0 ={}&
\tfrac{(k-1)(k+2)}{k(k+1)}\mathcal{K}^0 \mathcal{K}^2
-\tfrac{4}{k+2}\mathcal{K}^1 \mathcal{K}^1,\\
\mathcal{K}^1 \mathcal{K}^0 ={}&
\tfrac{k}{k+2}\mathcal{K}^0 \mathcal{K}^1,\\
\mathcal{K}^2 \mathcal{K}^1 ={}&
\tfrac{k-2}{k}\mathcal{K}^1 \mathcal{K}^2,\\
\mathcal{K}^1 \mathcal{K}^1 ={}&
\mathrm{Id} - \tfrac{k-1}{k}\mathcal{K}^0 \mathcal{K}^2,
\end{align}
\end{subequations}
where $\mathrm{Id}$ is the identity operator.
\end{lemma}

\section{GHP form}\label{sec:GHPform}
Given a spinor dyad $(o_A, \iota_A)$, any spinor $\varphi_{A_1 \dots A_k A'_1 \dots A'_l} \in \SymSpin_{k,l}$ can be represented in terms of its $(k+1)(l+1)$ Newman-Penrose scalars
\begin{align} \label{eq:GHPcompDef}
\varphi_{i j'} = \varphi_{A_1 \dots A_k A'_1 \dots A'_l} \iota^{A_1}\dots \iota^{A_i}o^{A_{i+1}}\dots o^{A_k} \bar{\iota}^{A'_1}\dots \bar{\iota}^{A'_j}\bar{o}^{A'_{j+1}}\dots \bar{o}^{A'_l}.
\end{align}
In particular the Weyl spinor $\Psi_{ABCD}$ corresponds to the five complex Weyl scalars $\Psi_0, \dots , \Psi_4$. The dyad normalization $o_A \iota^A= 1$ is invariant under the transformation $o_A \to \lambda o_A$,  $\iota_A \to \lambda^{-1} \iota_A$ with $\lambda$ a non-vanishing, complex scalar field. It follows that \eqref{eq:GHPcompDef} transforms as a section of a complex line bundle,
\begin{align}
\varphi_{i j'} \to \lambda^p \bar{\lambda}^q \varphi_{i j'}
\end{align}
with $p = k - 2i, q = l - 2j$ and it is said to be of type $\{p,q\}$. The Levi-Civita connection $\nabla_{AA'}$ lifted to the complex line bundles of weighted fields is denoted by $\Theta_{AA'}$ and its dyad components are the weighted GHP operators
\begin{subequations} 
\begin{align}
\tho &= o^A \bar{o}^{A'}\Theta_{AA'}, &
\tho' &= \iota^A \bar{\iota}^{A'}\Theta_{AA'}, \\
\edt &= o^A \bar{\iota}^{A'}\Theta_{AA'}, &
\edt' &= \iota^A \bar{o}^{A'}\Theta_{AA'}.
\end{align}
\end{subequations} 
The remaining (properly weighted) complex connection coefficients are denoted by $\rho, \rho', \tau, \tau'$, $ \kappa, \kappa', \sigma, \sigma'$ and only the first four are non-zero with respect to a principal tetrad on vacuum type D spacetimes. Background and more details about the GHP formalism can be found in the original work \cite{GHP}. In this section we collect the dyad components of various covariant operators introduced in previous sections.
\subsection{Spin-1}\label{sec:GHPformSpin1}
The components of the TME operator \eqref{eq:Spin1Odef} are
\begin{subequations} \label{eq:Spin1TMEGHPForm}
\begin{align}
(\TMEOop \phi)_{0}={}&
 -  \frac{\kappa_1}{\bar{\kappa}_{1'}} \tho \bigl(\kappa_1 \bar{\kappa}_{1'}(\tho' - 2 \rho')\bigr)\phi_{0} \nonumber \\
 &+ \frac{\kappa_1}{\bar{\kappa}_{1'}} \edt \bigl(\kappa_1 \bar{\kappa}_{1'}(\edt' - 2 \tau')\bigr)\phi_{0} + \tfrac{1}{3} \kappa_1 \mathcal{L}_{\xi}\phi_{0},\\
(\TMEOop \phi)_{1}={}&0,\\
(\TMEOop \phi)_{2}={}&
 -  \frac{\kappa_1}{\bar{\kappa}_{1'}} \tho' \bigl(\kappa_1 \bar{\kappa}_{1'}(\tho - 2 \rho)\bigr)\phi_{2} \nonumber \\
 &+ \frac{\kappa_1}{\bar{\kappa}_{1'}} \edt' \bigl(\kappa_1 \bar{\kappa}_{1'}(\edt - 2 \tau)\bigr)\phi_{2} - \tfrac{1}{3} \kappa_1 \mathcal{L}_{\xi}\phi_{2}.
\end{align}
\end{subequations} 
The components of the TSI operator \eqref{eq:Spin1Ohatdef} are
\begin{subequations} \label{eq:Spin1TSIGHPForm}
\begin{align}
(\TSIOop \phi)_{0'}={}&- \frac{\bar{\kappa}_{1'}}{\kappa_1} \tho \tho (\kappa_1^2\phi_{2})
 + \frac{\bar{\kappa}_{1'}}{\kappa_1} \edt' \edt' (\kappa_1^2\phi_{0}),\\
(\TSIOop \phi)_{1'}={}&0,\\
(\TSIOop \phi)_{2'}={}&- \frac{\bar{\kappa}_{1'}}{\kappa_1} \tho' \tho' (\kappa_1^2\phi_{0})
 + \frac{\bar{\kappa}_{1'}}{\kappa_1} \edt \edt (\kappa_1^2\phi_{2}).
\end{align}
\end{subequations} 
The components of  the linear symmetry operator \eqref{eq:Spin1CurlShatst} are
\begin{subequations} \label{eq:Spin1CurlShatstComps} 
\begin{align}
\psi_{0}={}& \tho \bigl(\kappa_1 \bar{\kappa}_{1'}(\tho'
 - 2 \rho')\bigr)\phi_{0}  
 + \edt \bigl(\kappa_1 \bar{\kappa}_{1'}(\edt'
 - 2 \tau')\bigr)\phi_{0} \nonumber\\
 & - \tfrac{1}{3} \kappa_1 \mathcal{L}_{\bar{\xi}}\phi_{0},\\
 \psi_{1}={}&\tfrac{1}{2} (\tho' + \rho')\bigl(\kappa_1 \bar{\kappa}_{1'}(\edt' - 2 \tau')\bigr)\phi_{0} \nonumber\\
  &+ \tfrac{1}{2} (\edt' + \tau')\bigl(\kappa_1 \bar{\kappa}_{1'}(\tho' - 2 \rho')\bigr)\phi_{0}\nonumber \nonumber\\
& + \tfrac{1}{2} (\tho + \rho)\bigl(\kappa_1 \bar{\kappa}_{1'}(\edt - 2 \tau)\bigr)\phi_{2} \nonumber\\
 & + \tfrac{1}{2} (\edt + \tau)\bigl(\kappa_1 \bar{\kappa}_{1'}(\tho - 2 \rho)\bigr)\phi_{2},\\
\psi_{2}={}& \tho' \bigl(\kappa_1 \bar{\kappa}_{1'}(\tho
 - 2 \rho)\bigr)\phi_{2}
 + \edt' \bigl(\kappa_1 \bar{\kappa}_{1'}(\edt
 - 2 \tau)\bigr)\phi_{2} \nonumber\\
 & + \tfrac{1}{3} \kappa_1 \mathcal{L}_{\bar{\xi}}\phi_{2}.
\end{align}
\end{subequations} 
The components of the anti-linear symmetry operator \eqref{eq:Spin1CurlDgSdg} are
\begin{subequations} \label{eq:Spin1CurlDgSdgComps}
\begin{align}
\bar{\chi}_{0'}={}&\tho \tho (\kappa_1^2\phi_{2})
 + \edt' \edt' (\kappa_1^2\phi_{0}),\\
\bar{\chi}_{1'}={}&(\tho \edt
 + \bar{\tau}'\tho )(\kappa_1^2\phi_{2})
 + (\tho' \edt'
 + \bar{\tau}\tho' )(\kappa_1^2\phi_{0}),\\
\bar{\chi}_{2'}={}&\tho' \tho' (\kappa_1^2\phi_{0})
 + \edt \edt (\kappa_1^2\phi_{2}).
\end{align}
\end{subequations} 
The Lie derivative of $\phi_{AB}$ components along $\xi$ is given by
\begin{subequations}\label{eq:GHPLieXiphi20} 
\begin{align}
\mathcal{L}_{\xi}\phi_{0}={}&
 - 3 \kappa_1 (\rho' \tho - \rho \tho' - \tau' \edt + \tau \edt' - \Psi_{2}) \phi_{0},\\
\mathcal{L}_{\xi}\phi_{1}={}&-3 \kappa_1 (\rho' \tho - \rho \tho' - \tau' \edt + \tau \edt') \phi_{1},\\
\mathcal{L}_{\xi}\phi_{2}={}& - 3 \kappa_1 (\rho' \tho - \rho \tho' - \tau' \edt + \tau \edt' + \Psi_{2} ) \phi_{2},
\end{align}
\end{subequations}

\subsection{Spin-2}\label{sec:GHPformSpin2}
The components of the TME operator \eqref{eq:Spin2OopDef} are
\begin{subequations} \label{eq:Spin2TMEGHPForm}
\begin{align}
(\TMEOop \phi)_{0}={}&
 -  \frac{\kappa_1^3}{\bar{\kappa}_{1'}} \tho \bigl(\kappa_1 \bar{\kappa}_{1'}(\tho' - 4 \rho')\bigr)\phi_{0} \nonumber \\
 &+ \frac{\kappa_1^3}{\bar{\kappa}_{1'}} \edt \bigl(\kappa_1 \bar{\kappa}_{1'}(\edt' - 4 \tau')\bigr)\phi_{0} + \kappa_1^3 \mathcal{L}_{\xi}\phi_{0},\\
(\TMEOop \phi)_{1}={}&0, \quad
(\TMEOop \phi)_{2}={}0,\quad
(\TMEOop \phi)_{3}={}0,\\
(\TMEOop \phi)_{4}={}&
 -  \frac{\kappa_1^3}{\bar{\kappa}_{1'}} \tho' \bigl(\kappa_1 \bar{\kappa}_{1'}(\tho - 4 \rho)\bigr)\phi_{4} \nonumber \\
 &+ \frac{\kappa_1^3}{\bar{\kappa}_{1'}} \edt' \bigl(\kappa_1 \bar{\kappa}_{1'}(\edt - 4 \tau)\bigr)\phi_{4} - \kappa_1^3 \mathcal{L}_{\xi}\phi_{4}.
\end{align}
\end{subequations} 
The components of the operator \eqref{eq:TSIOperator} are
\begin{subequations} \label{eq:Spin2TSIOpGHPForm}
\begin{align}
(\TSIOop \phi)_{0'}={}&- \frac{\bar{\kappa}_{1'}^3}{\kappa_1} \tho \tho \tho \tho (\kappa_1^4\phi_{4})
 + \frac{\bar{\kappa}_{1'}^3}{\kappa_1} \edt' \edt' \edt' \edt' (\kappa_1^4\phi_{0}),\\
(\TSIOop \phi)_{1'}={}&0,\quad
(\TSIOop \phi)_{2'}={}0,\quad
(\TSIOop \phi)_{3'}={}0,\\
(\TSIOop \phi)_{4'}={}&- \frac{\bar{\kappa}_{1'}^3}{\kappa_1} \tho' \tho' \tho' \tho' (\kappa_1^4\phi_{0})
 + \frac{\bar{\kappa}_{1'}^3}{\kappa_1} \edt \edt \edt \edt (\kappa_1^4\phi_{4}).
\end{align}
\end{subequations}

The Lie derivative of $\phi_{ABCD}$ components along $\xi$ is given by
\begin{subequations} \label{eq:GHPLieXiphi40}
\begin{align}
\mathcal{L}_{\xi}\phi_{0}={}&
 - 3 \kappa_1 (\rho' \tho -  \rho \tho'  -  \tau' \edt  +  \tau \edt' -2 \Psi_{2}) \phi_{0},\\
\mathcal{L}_{\xi}\phi_{1}={}& - 3 \kappa_1 (\rho' \tho - \rho \tho'  - \tau' \edt  + \tau \edt' - \Psi_{2}) \phi_{1},\\
\mathcal{L}_{\xi}\phi_{2}={}&-3 \kappa_1 (\rho' \tho - \rho \tho' - \tau' \edt + \tau \edt') \phi_{2},\\
\mathcal{L}_{\xi}\phi_{3}={}&
 - 3 \kappa_1 (\rho' \tho - \rho \tho' - \tau' \edt + \tau \edt' + \Psi_{2}) \phi_{3},\\
\mathcal{L}_{\xi}\phi_{4}={}&
 - 3 \kappa_1 (\rho' \tho - \rho \tho' - \tau' \edt + \tau \edt' + 2 \Psi_{2}) \phi_{4}.
\end{align}
\end{subequations} 

The curvature components of the Debye map given by the metric \eqref{eq:Spin2LinMetrichpm} for the plus sign are given by
\begin{subequations} \label{eq:Spin2LinMetrichpCurvGHP}
\begin{align}
\overline{\underline{\chi}}^{+}_{0'}={}&\edt' \edt' \edt' \edt' (\kappa_1^4\phi_{0}),\\
\overline{\underline{\chi}}^{+}_{1'}={}&\bigl(\tho' (\edt'
 -  \bar{\tau})
 + 3 \rho' \bar{\tau}
 - 6 \bar{\rho}' \bar{\tau}
 + 3 \bar{\rho}' \tau'\bigr)\nonumber \\
  &(\edt'
 + 2 \bar{\tau})(\edt'
 + 2 \bar{\tau})(\kappa_1^4\phi_{0}),\\
\overline{\underline{\chi}}^{+}_{2'}={}&\bigl( (\edt'
 -  \bar{\tau})\tho' + 6 \rho' \bar{\tau}
 - 12 \bar{\rho}' \bar{\tau}
 + 6 \bar{\rho}' \tau' \bigr)\nonumber \\
  &\bigl((\edt'
  + 2 \bar{\tau})(\tho'
  + 3 \bar{\rho}') + \rho' \bar{\tau}
 - 2 \bar{\rho}' \bar{\tau}
 + \bar{\rho}' \tau' \bigr)(\kappa_1^4\phi_{0}),\\
\overline{\underline{\chi}}^{+}_{3'}={}& \left(\edt' (\tho'
 -  \bar{\rho}')
 + 3 \rho' \bar{\tau}
 - 6 \bar{\rho}' \bar{\tau}
 + 3 \bar{\rho}' \tau'\right)\nonumber \\
  &(\tho'
 + 2 \bar{\rho}')(\tho'
 + 2 \bar{\rho}')(\kappa_1^4\phi_{0}),\\
\overline{\underline{\chi}}^{+}_{4'}={}&\tho' \tho' \tho' \tho' (\kappa_1^4\phi_{0}),\\
\underline{\psi}^{+}_{0}={}&\Psi_{2} \kappa_1^3 \mathcal{L}_{\xi}\phi_{0}, \\
 \underline{\psi}^{+}_{1}={}&0, \qquad
\underline{\psi}^{+}_{2}={}0,\qquad
\underline{\psi}^{+}_{3}={}0,\qquad
\underline{\psi}^{+}_{4}={}0,
\end{align}
\end{subequations}
and for the minus sign we get
\begin{subequations} \label{eq:Spin2LinMetrichmCurvGHP}
\begin{align}
\overline{\underline{\chi}}^{-}_{0'}={}&\tho \tho \tho \tho (\kappa_1^4\phi_{4}),\\
\overline{\underline{\chi}}^{-}_{1'}={}& \bigl(\edt (\tho
 -  \bar{\rho})
 + 3 \bar{\rho} \tau
 + 3 \rho \bar{\tau}'
 - 6 \bar{\rho} \bar{\tau}' \bigr)\nonumber \\
  &(\tho
 + 2 \bar{\rho})(\tho
 + 2 \bar{\rho})(\kappa_1^4\phi_{4}),\\
\overline{\underline{\chi}}^{-}_{2'}={}&\bigl( (\edt
 -  \bar{\tau}')\tho + 6 \bar{\rho} \tau
 + 6 \rho \bar{\tau}'
 - 12 \bar{\rho} \bar{\tau}' \bigr) \nonumber \\
 &\bigl((\edt
  + 2 \bar{\tau}')(\tho
  + 3 \bar{\rho}) + \bar{\rho} \tau
 + \rho \bar{\tau}'
 - 2 \bar{\rho} \bar{\tau}' \bigr)(\kappa_1^4\phi_{4}),\\
\overline{\underline{\chi}}^{-}_{3'}={}& \bigl(\tho (\edt -\bar{\tau}')
 + 3 \bar{\rho} \tau
 + 3 \rho \bar{\tau}'
 - 6 \bar{\rho} \bar{\tau}' \bigr)\nonumber \\
  &(\edt
 + 2 \bar{\tau}')(\edt
 + 2 \bar{\tau}')(\kappa_1^4\phi_{4}),\\
\overline{\underline{\chi}}^{-}_{4'}={}&\edt \edt \edt \edt (\kappa_1^4\phi_{4}),\\
\underline{\psi}^{-}_{0}={}&0,\qquad
\underline{\psi}^{-}_{1}={}0,\qquad
\underline{\psi}^{-}_{2}={}0,\qquad
\underline{\psi}^{-}_{3}={}0, \\
\underline{\psi}^{-}_{4}={}&- \Psi_{2} \kappa_1^3 \mathcal{L}_{\xi}\phi_{4}.
\end{align}
\end{subequations} 
Restricting to the real or imaginary part of the metric \eqref{eq:Spin2LinMetrichpm} leads to the linear combinations $\frac{1}{2}(\overline{\underline{\chi}}^{\pm}_{n} + \overline{\underline{\psi}}^{\pm}_{n})$ or $\frac{1}{2i}(\overline{\underline{\chi}}^{\pm}_{n} - \overline{\underline{\psi}}^{\pm}_{n})$ for the self-dual curvature.

We note that on vacuum type D background the identities
\begin{subequations} 
\begin{align*}
\edt' \edt' \edt' \edt' (\kappa_1^4\phi_{0})={}&(\edt'
 -  \bar{\tau})(\edt'
 -  \bar{\tau})(\edt'
 -  \bar{\tau})(\edt'
 + 3 \bar{\tau})(\kappa_1^4\phi_{0}), \\
\tho' \tho' \tho' \tho' (\kappa_1^4\phi_{0})={}&(\tho'
 -  \bar{\rho}')(\tho'
 -  \bar{\rho}')(\tho'
 -  \bar{\rho}')(\tho'
 + 3 \bar{\rho}')(\kappa_1^4\phi_{0}),
\end{align*}
\end{subequations} 
and their GHP-primed and c.c. versions follow from the Ricci identities. This connects the more compact form presented here to the equations given in \cite{wald:1978PhRvL..41..203W}, \cite{1975PhLA...54....5C}, \cite{Chrzanowski:1975}.

\subsection*{Acknowledgements} We are grateful to Lars Andersson for many enlightening discussions and remarks. We also thank Bernard Whiting for his interest in this work.


%

\end{document}